\documentclass[11pt]{article}
\usepackage{amsmath}
\usepackage{amssymb}

\usepackage[english]{babel}
\usepackage{latexsym}
\usepackage{amstext, amsfonts,amsbsy}

\usepackage{colortbl}
\usepackage{longtable}

\catcode`\@=11
\def\marginnote#1{}

        \def\theequation{\thesection.\arabic{equation}}

\newcommand{\tr}{{\rm tr}}
\newcommand{\ti}[1]{\tilde{#1}}

\newcommand{\mL}{{\mathcal L}}
\newcommand{\mM}{{\mathcal M}}

\newcommand{\mH}{{\mathcal H}}

\newcommand{\mO}{{\mathcal O}}

\newcommand{\vf}{\varphi}
\newcommand{\al}{\alpha}
\newcommand{\be}{\beta}
\newcommand{\ga}{\gamma}
\newcommand{\om}{\omega}
\newcommand{\vth}{\vartheta}
\newcommand{\de}{\delta}

\newcommand{\Mat}{ {\rm Mat}(N,\mathbb C) }

\newcommand{\MatNM}{ {\rm Mat}(NM,\mathbb C) }

\newcommand{\mC}{\mathbb C}
\newcommand{\mZ}{\mathbb Z}

\newcommand{\mS}{\mathcal S}

\newcommand{\ka}{\kappa}

\newtheorem{theorem}{Theorem}

\newenvironment{proof}{\par\noindent{\bf Proof.}}{\hfill$\scriptstyle\blacksquare$}

\def\beq{\begin{equation}}
\def\eq{\end{equation}}

\def\res{\mathop{\hbox{Res}}\limits}

\begin{document}

\setcounter{page}{1}

\

\vspace{-15mm}

\begin{flushright}
\end{flushright}
\vspace{0mm}

\begin{center}
\vspace{-8mm}
{\LARGE{ Lax equations for relativistic ${\rm GL}(NM,{\mathbb C})$}}
 \\ \vspace{4mm}
{\LARGE{ Gaudin models on elliptic curve}}

 \vspace{16mm}

 {\Large  {E. Trunina}\,\footnote{Steklov Mathematical Institute of Russian
Academy of Sciences, Gubkina str. 8, 119991, Moscow, Russia;
\\
 Moscow
 Institute of Physics and Technology, Inststitutskii per.  9,
 Dolgoprudny, Moscow region, 141700, Russia;
 e-mail:
 yelizaveta.kupcheva@phystech.edu.}
 \quad\quad\quad
 {A. Zotov}\,\footnote{Steklov Mathematical Institute of Russian
Academy of Sciences, Gubkina str. 8, 119991, Moscow, Russia;
 e-mail: zotov@mi-ras.ru.}
 }
\end{center}

\vspace{5mm}

\begin{abstract}
 We describe the most general ${\rm GL}_{NM}$ classical elliptic finite-dimensional integrable system, which
 Lax matrix has $n$ simple poles on elliptic curve. For $M=1$ it reproduces the classical
  inhomogeneous spin chain, for $N=1$ it is the Gaudin type (multispin) extension of the spin Ruijsenaars-Schneider model,
   and for $n=1$ the model of $M$ interacting relativistic ${\rm GL}_N$ tops emerges in some particular case.
   In this way we present a classification for relativistic Gaudin models on ${\rm GL}$-bundles over elliptic curve.
  As a by-product we describe the inhomogeneous Ruijsenaars chain. We show that this model can be considered as a particular case of multispin Ruijsenaars-Schneider model when residues of the Lax matrix are of rank one. An explicit parametrization of the classical spin variables through the canonical variables is obtained for this model.
  Finally, the most general ${\rm GL}_{NM}$ model is also described through $R$-matrices satisfying associative Yang-Baxter equation. This description provides the trigonometric and rational analogues of ${\rm GL}_{NM}$ models.
\end{abstract}




\small{
\tableofcontents
}
\parskip 5pt plus 1pt   \jot = 1.5ex


\section{Introduction: classification scheme}
\setcounter{equation}{0}

In our previous paper \cite{TZ} we reviewed the non-relativistic classical integrable systems
on elliptic curve. The classification scheme for these model is as follows:
  $$
   \begin{array}{c}
   \hbox{\underline{Classification scheme for elliptic non-relativistic models:}}
   \\ \ \\
   \begin{array}{ccc}
      & \fbox{ $\phantom{\Big(}$\quad 1. general ${\rm gl}_{NM}^{\times n}$ model\quad $\phantom{\Big(}$} &
   \\
  \hfill$\qquad\qquad\qquad\ \phantom{\Big(}$ \hbox{\footnotesize{$M=1$}}\swarrow &  \Big| & \searrow
  \hbox{\footnotesize{$N=1$}}\quad $\phantom{\Big(}$\qquad\qquad\qquad\qquad\
      \\
    \hfill\fbox{$\phantom{\Big(}$ 2. ${\rm gl}_N^{\times n}$ Gaudin model$\phantom{\Big(}$} &  \stackrel{ \hbox{\footnotesize{$n=1$}} }{\downarrow} &
    \fbox{$\phantom{\Big(}$ 3. ${\rm gl}_M^{\times n}$ multispin CM}\hfill
         \\
   \Big| & \fbox{$\phantom{\Big(}$ 4. ${\rm gl}_{NM}$ mixed type model$\phantom{\Big(}$} & \Big|\hfill
   \\
  $\qquad\qquad\ \ \phantom{\Big(}$ \stackrel{ \hbox{\footnotesize{$n=1$}} }{\downarrow}\ \, \hfill\hbox{\footnotesize{$M=1$}}\swarrow &  \Big| & \searrow
  \hbox{\footnotesize{$N=1$}}\quad\ \  $\phantom{\Big(}$\stackrel{ \hbox{\footnotesize{$n=1$}} }{\downarrow}
  $\qquad\qquad\quad\ \phantom{\Big(}$\
      \\
    \hfill\fbox{$\phantom{\Big(}$ 5. ${\rm gl}_N$ integrable top\ $\phantom{\Big(}$} &  $\phantom{\Big(}$
    \stackrel{ \hbox{\footnotesize{$rk(S)=1$}} }{\downarrow} $\phantom{\Big(}$ &
    \fbox{$\phantom{\Big(}$ 6.  ${\rm gl}_M$ spin CM\qquad $\phantom{\Big(}$}\hfill
         \\
   \ \Big| &  \fbox{$\phantom{\Big(}$ 7.  $M$ interacting ${\rm gl}_N$ tops $\phantom{\Big(}$}
  &
  \ \Big|
            \\
 $\qquad\qquad\ \phantom{\Big(}$\hfill\stackrel{ \hbox{\footnotesize{$rk(S)=1$}} }{\downarrow}\ \hbox{\footnotesize{$M=1$}}  \swarrow &
  &
   \searrow
  \hbox{\footnotesize{$N=1$}}\quad\ \stackrel{ \hbox{\footnotesize{$rk(S)=1$}}}{\downarrow}$\qquad\qquad\ \phantom{\Big(}$
            \\
    \hfill\fbox{$\phantom{\Big(}$ 8. ${\rm gl}_N$ top on $\mO_N^{\hbox{\tiny{min}}}$ \quad $\phantom{\Big(}$} & &
    \fbox{$\phantom{\Big(}$ 9. ${\rm gl}_M$ spinless CM $\phantom{\Big(}$}
    \\ \ \\
    \hfill\hbox{family II}\qquad\qquad   &  \hbox{family III} & \qquad  \hbox{family I}\qquad\qquad
    \\ \ \\
       & \hbox{\bf Scheme 1} &
   \end{array}
      \end{array}
   $$
%
%
%
%
%
%
%
\paragraph{Non-relativistic models.}  Let us briefly recall the main idea. The lowest level is given by the elliptic spinless Calogero-Moser (CM) model
and the elliptic top with minimal coadjoint orbit. These are the boxes 9 and 8 on the Scheme 1 respectively.
Within the first family the Calogero-Moser model is extended to its spin generalization (box 6) and the
Nekrasov's multispin model 3 of Gaudin type. Similarly, in the second family the elliptic top with minimal coadjoint orbit is extended
to the one with arbitrary orbit (box 5) and to the elliptic Gaudin model (box 2). Hereinafter by the models of Gaudin
type we mean those models, which are described by the Lax matrices with a set of simple poles in spectral parameter $z$ at some points $z_1,...,z_n$ on elliptic curve (or its degeneration). In the classical spin Calogero-Moser model the spin variables  are arranged into the ''spin matrix'' $S$, which is a residue of the Lax matrix at a single pole. In the Gaudin models there are $n$ spin matrices -- residues at $z_1,...,z_n$. For this reason the Gaudin type models are also called as multi-pole or multispin models. Below we keep this terminology for the relativistic models.

The models from the second family are governed by the classical non-dynamical $r$-matrices of vertex type, while the systems from the first family
are described by dynamical (classical analogues of IRF type) $r$-matrices. According to classification of classical elliptic integrable systems \cite{LOSZ} there are also intermediate ${\rm gl}_{NM}$ models of mixed type. They are presented in
the third family.  When $N=1$ the first family is reproduced, and the second family appears in the case $M=1$.
The models from different families are related by the so-called symplectic Hecke correspondence \cite{LOZ}. In particular, it means that the systems 8 and 9 are gauge equivalent at the level of Lax pairs, and explicit change of variables can be evaluated.

The upper level of the Scheme 1 (i.e. the cases 1, 2, 3) is given by the Gaudin type models. In these cases the Lax matrices have $n$ simple poles. When $n=1$ these models turn into the middle level (i.e. the cases 4, 5, 6). And the lower level (i.e. the cases 7, 8, 9) comes from the middle one by restricting to the coadjoint orbits
of minimal dimensions for the spin variables. The spin variables are elements of the spin matrix $S$, and the condition ${\rm rk}(S)=1$ is equivalent to the choice of minimal coadjoint orbit.

\paragraph{Relativistic models.} In this paper we discuss relativistic analogue of the above scheme. The classification is presented on the Scheme 2.
%
%
  $$
   \begin{array}{c}
   \hbox{\underline{Classification scheme for elliptic relativistic models:}}
   \\ \ \\
   \begin{array}{ccc}
      & \fbox{ $\phantom{\Big(}$ 1. general ${\rm GL}_{NM}^{\times n}$ model $\phantom{\Big(}$} &
   \\
  \hfill$\qquad\qquad\qquad\phantom{\Big(}$ \hbox{\footnotesize{$M=1$}}\swarrow &  \Big| & \searrow
  \hbox{\footnotesize{$N=1$}} $\phantom{\Big(}$\qquad\qquad\qquad\qquad\
      \\
    \hfill\fbox{$\phantom{\Big|}$ 2. ${\rm GL}_N^{\times n}$ XYZ chain\quad $\phantom{\Big|}$} &  \stackrel{ \hbox{\footnotesize{$n=1$}} }{\downarrow} &
    \fbox{ 3. ${\rm GL}_M^{\times n}$ multispin RS\phantom{\Big|}}\hfill
         \\
   \Big| & \fbox{$\phantom{\Big|}$ 4. ${\rm GL}_{NM}$ mixed type model$\phantom{\Big|}$} & $\phantom{ \hbox{\footnotesize{$n=1$}} }$\Big|\qquad\
   \\
  $\qquad\qquad\phantom{\Big(}$ \stackrel{ \hbox{\footnotesize{$n=1$}} }{\downarrow} \hfill\hbox{\footnotesize{$M=1$}}\swarrow &  \Big| & \searrow
  \hbox{\footnotesize{$N=1$}}\quad  $\phantom{\Big(}$\stackrel{ \hbox{\footnotesize{$n=1$}} }{\downarrow}
  \qquad\qquad\ \
      \\
    \hfill\fbox{$\phantom{\Big(}$ 5. ${\rm GL}_N$ relativ. top\ $\phantom{\Big(}$} &  $\phantom{\Big(}$
    \stackrel{ \hbox{\footnotesize{$rk(S)=1$}} }{\downarrow} $\phantom{\Big(}$ &
    \fbox{$\phantom{\Big(}$ 6.  ${\rm GL}_M$ spin RS\qquad $\phantom{\Big(}$}\hfill
         \\
   \Big|  &  \fbox{$\phantom{\Big(}$ 7.  $M$ interacting ${\rm GL}_N$ tops $\phantom{\Big(}$}
  &
   \Big|
            \\
 \qquad\qquad\ \stackrel{ \hbox{\footnotesize{$rk(S)=1$}} }{\downarrow}\! \hbox{\footnotesize{$M=1$}}  \swarrow &
  &
   \searrow
  \hbox{\footnotesize{$N=1$}} \stackrel{ \hbox{\footnotesize{$rk(S)=1$}} }{\downarrow}\qquad\qquad\
            \\
    \hfill\fbox{$\phantom{\Big(}$ 8. special ${\rm GL}_N$ top $\phantom{\Big(}$} & &
    \fbox{$\phantom{\Big(}$ 9. ${\rm GL}_M$ spinless RS $\phantom{\Big(}$}\hfill
    \\ \ \\
    \hfill\hbox{family II}\qquad\qquad   &  \hbox{family III} & \qquad\qquad  \hbox{family I}\qquad\qquad\
    \\ \ \\
       & \hbox{\bf Scheme 2} &
   \end{array}
      \end{array}
   $$
Consider {\bf the first family}. The relativistic (spinless) many-body system is the elliptic Ruijsenaars-Schneider model \cite{Ruij} (box 9). It was extended to the spin case
by Krichever and Zabrodin \cite{KrichZ} (box 6). A generalization of the latter to the multi-pole case (box 3)
is also known in the literature. See, for example,  \cite{MMZ},
where such model appears in the context of dualities.
At the same time, the classical equations of motion and/or the Hamiltonian description were not known to our best knowledge. In fact,  the multispin  Ruijsenaars-Schneider is known much better at quantum level.
It is constructed by means of dynamical Felder's elliptic $R$-matrix \cite{Felder}.
Such models are also studied in the context of supersymmetric gauge theories and underlying Seiberg-Witten geometry, see e.g. \cite{MOY,MMZ} and references therein.

The quantization problem is also related to a known open problem -- to describe the Poisson and $r$-matrix structures for the spin elliptic Ruijsenaars-Schneider model. The Hamiltonian description is still unknown. This is why we discuss the Lax equations only in the general case\footnote{There is no a full proof of integrability for elliptic spin relativistic models since the classical r-matrix structure is unknown. However, there are some arguments for integrability besides existence of the Lax pair. On one hand there is a quantum RLL algebra \cite{SeZ2,SeZ3}, and on the other hand commutativity of anisotropic spin Ruijsenaars Hamiltonians was proved directly at quantum level in \cite{MZ}.}.
  At the same time much progress was achieved in the studies of  trigonometric spin Ruijsenaars-Schneider models, see \cite{AO,ChF,EPestun,Feher}.
 Although we do not address to precisely this problem, we derive explicit parametrization of spin variables through canonical variables in a special case of model 3, when all spin matrices are of rank one.

The models from {\bf the second family} are the classical analogues of XYZ spin chain including the higher rank generalizations \cite{Skl2,FTbook}. The model 2 is the ${\rm GL}_N$ inhomogeneous classical XYZ spin chain on $n$ sites. The models 5 and 8 can be viewed as 1-site classical chain with the Poisson structure given by the classical Sklyanin algebra \cite{Skl1}. From viewpoint of the classical mechanics these models are relativistic tops \cite{LOZ8}. The model 8 is a particular case of model 5 corresponding to the special case, when the matrix of spin variables $S$
has rank one.

Finally, {\bf the third family} consists of the mixed type ${\rm GL}_{NM}$ models similarly to its non-relativistic analogue from the Scheme 1.
The models 4 and 7 on the Scheme 2 were described in \cite{Reltops,ReltopsR}. The model 1 is on the top of the Scheme 2, and this is the subject of this article. Similarly to non-relativistic case the families on the Scheme 2 are related by the symplectic Hecke correspondence. For example, the models 8 and 9 are gauge equivalent. This phenomenon was originally observed by K. Hasegawa \cite{Hasegawa}, see also \cite{Chen,VZ,KZ19,ZZ}.

The study of models from the third family can be interesting from different viewpoints. Their quantum analogues are described by the mixed type quantum $R$-matrices \cite{LOSZ2}, which turn into the vertex type when $M=1$ and become of purely IRF type in the case $N=1$. The underlying quantum algebra takes the form of an intermediate case between the
Sklyanin algebra and the elliptic quantum group \cite{SeZ2}. Recently the quantum Hamiltonians for ${\rm GL}_{NM}$ model of interacting tops (box 7) were proposed and used for construction of new long-range spin chains \cite{MZ}.
The studies of multi-pole case in this context is an interesting open problem.

Another possible application
of the Gaudin type models arises in the studies of 1+1 integrable field theories generalizing the models on the Schemes 1 and 2. The 1+1 generalizations of the Calogero-Moser system is known \cite{Krichever02} as well as the continuous models
of the Heisenberg-Landau-Lifshitz type \cite{Skl3,FTbook}. The 1+1 version of the spin and multispin Calogero-Moser models (the box 3 on the Scheme 1) was given in
 \cite{LOZ} and the 1+1 Gaudin models generalizing the models 2 on the Scheme 1 were discussed in \cite{Z11,AtaZ}. The field generalizations of Hitchin systems including the multi-pole type models are actively studied nowadays \cite{Vicedo,LOZ22}.
  At the relativistic level the
 1+1 field theories corresponding to the models 5 on the Scheme 2 are known from \cite{DNY15}. Recently the 1+1 analogue of the Ruijsenaars-Schneider model was suggested in \cite{ZZ}. It is an interesting classification problem to
 describe the soliton equations related to all boxes on the Schemes 1 and 2.

\paragraph{Purpose of the paper} is to present the classification Scheme 2 and describe the most general model 1. The generalized version of this model is also proposed through $R$-matrix formulation, which includes trigonometric and rational degenerations of the elliptic model. We also suggest explicit parametrization of the reduced multispin Ruijsenaars-Schneider model with rank one matrices of the spin variables.

The paper is organized as follows. In Section \ref{sec2} we review the models from the family II and recall the classical IRF-Vertex relation between the special relativistic top and the spinless Ruijsenaars-Schneider model.
The monodromy matrices of spin chains are naturally represented in the additive form, which is similar to the one for (non-relativistic) Gaudin models. However, in contrast to non-relativistic case, where the underlying quantum or classical algebra of exchange relations is based on Lie algebra, in the relativistic case we deal with quadratic algebras of Sklyanin type. The term relativistic Gaudin model is understood as a model with some multi-pole (and multispin) Lax matrix
and (possibly complicated or, even more, unknown) quadratic Poisson structure.
In Section \ref{sec3} the most general elliptic model is described. Namely, a natural ansatz for the Lax pair is suggested and the equations of motion are derived. In Section \ref{sec4} we propose inhomogeneous generalization of the Ruijsenaars spin chain. It is obtained by the gauge transformation of IRF-Vertex type starting from ${\rm GL}_N$ XYZ spin chain. As a result, we express the spin variables in the reduced multispin Ruijsenaars-Schneider model (with rank one spin matrices) through the set of canonical variables, thus providing the Hamiltonians description for this model.
Finally, in Section \ref{sec5} we extend the results of Section \ref{sec3} to $R$-matrix formulation based on the
associative Yang-Baxter equation.


\section{Quantum $R$-matrices in quantum and classical models}
\label{sec2}
\setcounter{equation}{0}

In this Section we introduce necessary notations and recall some basic facts and definitions by considering
the model 2 from the Scheme 2 -- ${\rm GL}_N$ spin chain on $n$ sites governed by the vertex type $R$-matrix.
A detailed description for the additive form of the monodromy matrices is given.

 \subsection{Quantum $R$-matrices and Yang-Baxter equations}
 A quantum $R$-matrix in the fundamental representation of ${\rm GL}_N$ Lie group is
 some $\Mat^{\otimes 2}$-valued function
 $$
 R^\hbar_{12}(q_1,q_2)\in\Mat^{\otimes 2}
 $$
  depending on the Planck constant $\hbar$ and the spectral
 parameters $q_1,q_2$. In fact, we assume that $R^\hbar_{12}(q_1,q_2)=R^\hbar_{12}(q_1-q_2)$, and the $R$-matrix
 is the elliptic ${\rm GL}_N$ Baxter-Belavin's one (\ref{BB}) or some its degeneration.
 In the general case any $\Mat^{\otimes 2}$-valued  $R$-matrix is of the following form:
  \beq\label{w01}
  \begin{array}{c}
    \displaystyle{
  R^\hbar_{12}(q)=\sum\limits_{ijkl=1}^N R_{ij,kl}(\hbar,q)E_{ij}\otimes E_{kl}\,,
 }
 \end{array}
 \eq
 where $\{E_{ij}\,;\, i,j=1...N\}$ is the standard matrix basis in $\Mat$, and $R_{ij,kl}(\hbar,q)$ is a set of functions. By definition any quantum $R$-matrix satisfies
 the quantum Yang--Baxter equation:
  \beq\label{w02}
    \displaystyle{
  R_{12}^\hbar(q_{12})R_{13}^\hbar(q_{13})R_{23}^\hbar(q_{23})=
R_{23}^\hbar(q_{23})R_{13}^\hbar(q_{13})R_{12}^\hbar(q_{12})\,,\quad q_{ij}=q_i-q_j\,,
 }
  \eq
where all $R$-matrices are considered as elements of $\Mat^{\otimes 3}$. For example,
  \beq\label{w03}
    \displaystyle{
  R^z_{12}(q)=\sum\limits_{ijkl=1}^N R_{ij,kl}(z,q)E_{ij}\otimes E_{kl}\otimes 1_N\,,\quad
  R^z_{13}(q)=\sum\limits_{ijkl=1}^N R_{ij,kl}(z,q)E_{ij}\otimes 1_N\otimes E_{kl}\,,
  }
  \eq
where $1_N$ is the identity matrix in $\Mat$. The elliptic $R$-matrix (\ref{BB}) satisfies (\ref{w03})
and the unitarity property
\begin{equation}\label{unitarity}\begin{array}{c}
    R^{z}_{12} (x) R^{z}_{21} (-x) = (\wp (z) - \wp (x))\, 1_N \otimes 1_N\stackrel{(\ref{diffsign})}{=}
    \phi(z,x)\phi(z,-x)\, 1_N \otimes 1_N\,,
\end{array}\end{equation}
 where $\wp(x)$ is the Weierstrass elliptic function, and $\phi(z,x)$ is the elliptic Kronecker function
 (\ref{philimits}).
  One more useful property of (\ref{BB}) is the skew-symmetry:
  \beq\label{w031}
    \displaystyle{
  R^z_{12}(q)=-R^{-z}_{21}(-q)\,.
  }
  \eq

Besides the quantum Yang-Baxter equation (\ref{w02}) the elliptic Baxter-Belavin $R$-matrix in the fundamental representation of the ${\rm GL}_N$ Lie group satisfies also the so-called associative Yang-Baxter equation (AYBE)
\cite{Pol}:
\begin{equation} \begin{array}{c}
\label{AYB}
    R^{z}_{12} R^{w}_{23} = R^{w}_{13} R^{z-w}_{12} + R^{w-z}_{23} R^{z}_{13}, \quad R^u_{ab} = R^u_{ab}(q_a-q_b),
\end{array} \end{equation}
 In contrast to (\ref{w02}) the latter equation remains nontrivial in the scalar case (when $N=1$). In this case
 it turns into the genus one Fay identity (\ref{Fay}), while the $R$-matrix itself becomes the elliptic Kronecker function (\ref{philimits}). Being a solution of the Yang-Baxter equation (\ref{w02}) an $R$-matrix is fixed up
 to multiplication by an arbitrary function. But this freedom is fixed in (\ref{AYB}), and the way of fixation
 is given by the r.h.s. of the unitarity property (\ref{unitarity}). More properties of the $R$-matrices under consideration can be found in the Appendix B and in \cite{LOZ15}.

\subsection{Quantum models}

 Let us recall some details on description of vertex type models from the second family since we use it throughout the paper. As a by-product we introduce necessary notations. At quantum level the spin chain (i.e. the model 2 from the Scheme 2) is described by means of a quantum $R$-matrix.

 \paragraph{Quantum spin chains.}
 The quantum inhomogeneous ${\rm GL}_N$ spin chain\footnote{Hereinafter we assume the closed spin chains only.} is defined by the monodromy matrix
  \beq\label{w04}
  \begin{array}{c}
    \displaystyle{
 {\hat T}(z)=R_{01}^\hbar(z-z_1)R_{02}^\hbar(z-z_2)\ldots R_{0n}^\hbar(z-z_n)\in\Mat\otimes{\rm End}(\mH)\,,
 }
 \end{array}
 \eq
 where $0$ denotes the auxiliary space $\Mat$, and indices $1,..,n$ are tensor components of the (quantum) Hilbert space $\mH$.
 If all $R$-matrices are in the fundamental representation of ${\rm GL}_N$ then $\mH=(\mC^N)^{\otimes n}$ and
 ${\rm End}(\mH)=\Mat^{\otimes n}$.
 Alternatively, one writes the monodromy matrix
  \beq\label{w05}
  \begin{array}{c}
    \displaystyle{
 {\hat T}(z)={\hat L}^{1}(z-z_1){\hat L}^{2}(z-z_2)\ldots {\hat L}^{n}(z-z_n)\,.
 }
 \end{array}
 \eq
 Each Lax operator ${\hat L}^i(z-z_i)\in\Mat$ is $N\times N$ matrix, which entries are operators acting on $\mH$.
 More precisely,
  \beq\label{w06}
  \begin{array}{c}
    \displaystyle{
 {\hat L}^{i}(z-z_i)={\hat L}^\hbar({\hat S}^{i},z-z_i)\,,\qquad {\hat S}^{i}=\sum\limits_{a,b=1}^N E_{ab} {\hat S}^{i}_{ab}
 }
 \end{array}
 \eq
 and\footnote{The standard notations are used: $A_1=A\otimes 1_N$ and $A_2=1_N\otimes A$ for any matrix $A\in\Mat$.}
  \beq\label{w07}
  \begin{array}{c}
    \displaystyle{
 {\hat L}^\hbar({\hat S}^{i},z-z_i)=\tr_2\Big(R^\hbar_{12}(z-z_i){\hat S}^{i}_2\Big)\stackrel{(\ref{w01})}{=}
 \sum\limits_{a,b,c,d=1}^N R_{abcd}(\hbar,z-z_i) E_{ab}\,{\hat S}^{i}_{dc}\,.
 }
 \end{array}
 \eq
 The commutation relations of quantum algebra between the operators ${\hat S}^i_{ab}$, $i=1,...,n$, $a,b=1,...,N$ are
 generated by $[{\hat S}^i_1,{\hat S}^j_2]=0$ (or, equivalently $[{\hat L}^{i}_1(z),{\hat L}^{j}_2(w)]=0$) for $i\neq j$ and the Sklyanin algebra \cite{Skl1}, which is a set of quadratic relations coming from the quantum exchange
  relations\footnote{The relations (\ref{w08}) are assumed to hold identically in spectral parameters $z$ and $w$. Hence (\ref{w08}) provides $N^4$ relations in the general case.}:
  \beq\label{w08}
  \begin{array}{c}
    \displaystyle{
 R^\hbar_{12}(z-w){\hat L}^\hbar_1({\hat S}^{i},z){\hat L}^\hbar_2({\hat S}^{i},w)=
{\hat L}^\hbar_2({\hat S}^{i},w){\hat L}^\hbar_1({\hat S}^{i},z)R^\hbar_{12}(z-w)\,.
 }
 \end{array}
 \eq
 It follows from these commutation relations that the monodromy matrix (\ref{w05}) also satisfies
 (\ref{w08}):
  \beq\label{w09}
  \begin{array}{c}
    \displaystyle{
 R^\hbar_{12}(z-w){\hat T}_1(z){\hat T}_2(w)=
{\hat T}_2(w){\hat T}_1(z)R^\hbar_{12}(z-w)\,.
 }
 \end{array}
 \eq
 Therefore\footnote{In fact, here we also use invertibility of $R^\hbar_{12}(z-w)$. It is true in our case due to the unitarity property (\ref{unitarity}).}, the quantum transfer-matrix
  \beq\label{w10}
  \begin{array}{c}
    \displaystyle{
 {\hat t}(z)=\tr_0{\hat T}(z)
 }
 \end{array}
 \eq
 obeys the property
  \beq\label{w11}
  \begin{array}{c}
    \displaystyle{
 [{\hat t}(z),{\hat t}(w)]=0\,.
 }
 \end{array}
 \eq
 It is an essential idea underlying the quantum inverse scattering method since it means that
 ${\hat t}(z)$ is a generating function of commuting Hamiltonians ${\hat H}_i$, $i=1,...,n$ (i.e. $[{\hat H}_i,{\hat H}_j]=0$), which can be defined as
  \beq\label{w111}
  \begin{array}{c}
    \displaystyle{
 {\hat H}_i=\res\limits_{z=z_i}{\hat t}(z)\,.
 }
 \end{array}
 \eq
 In order to write ${\hat H}_i$ explicitly we use that the residue of ${\hat L}({\hat S}^i,z-z_i)$ at point $z=z_i$ equals ${\hat S}^i$ (see (\ref{w131}) below). Then we have
  \beq\label{w112}
  \begin{array}{c}
    \displaystyle{
 {\hat H}_i=\tr \Big( {\hat L}^\hbar({\hat S}^1,z_i-z_1)\ldots {\hat L}^\hbar({\hat S}^{i-1},z_i-z_{i-1})\cdot
 {\hat S}^{i}\cdot {\hat L}^\hbar({\hat S}^{i+1},z_i-z_{i+1})\ldots {\hat L}^\hbar({\hat S}^{n},z_i-z_{n}) \Big)\,.
 }
 \end{array}
 \eq
 Alternatively, one can calculate ${\hat H}_i$ from (\ref{w04}) in the fundamental representation. Using the property (\ref{serRx}) one finds
 from the definitions (\ref{w10}) and (\ref{w111}) that
  \beq\label{w113}
  \begin{array}{c}
    \displaystyle{
 {\hat H}_i=R^\hbar_{i,i+1}(z_i-z_{i+1})\ldots R^\hbar_{i,n}(z_i-z_{n})\cdot
            R^\hbar_{i,1}(z_i-z_{1})\ldots R^\hbar_{i,i-1}(z_i-z_{i-1})\,.
 }
 \end{array}
 \eq

 It is also important to mention that the Sklyanin algebra generated by (\ref{w08}) has the fundamental representation of ${\rm GL}_N$ Lie group
  \beq\label{w12}
  \begin{array}{c}
    \displaystyle{
{\hat S}^i_{ab}=1_N\otimes...\otimes 1_N\otimes E_{ba}\otimes 1_N\otimes...\otimes 1_N\in\Mat^{\otimes n}\,,
 }
 \end{array}
 \eq
 where $E_{ba}$ is in the $i$-th tensor component. The representation (\ref{w12}) exists because in this case
 the Lax operators (\ref{w07}) ${\hat L}({\hat S}^{i},z-z_i)$ turn into $R$-matrices $R_{0i}^\hbar(z-z_i)$
 in the fundamental representation. The exchange relations (\ref{w08}) are then fulfilled due to
 the Yang-Baxter equation (\ref{w02}).

 It is also known \cite{Skl1,Hasegawa} that ${\rm GL}_N$ Sklyanin algebra has representation in terms of difference operators
 in $N$ variables. This case is the quantum analogue of the model 8 on the Scheme 2. It is closely related to the quantum Ruijsenaars-Schneider model (the model 9 on the Scheme 2) \cite{Hasegawa}.

\paragraph{Elliptic $L$-operator.} Let us write down explicit form of the elliptic Lax operator \cite{Skl1} using our
notation (\ref{w07}). Plugging the expression for elliptic $R$-matrix (\ref{BB}) into (\ref{w07}) and using (\ref{TrT}) one gets
  \beq\label{w13}
  \begin{array}{c}
    \displaystyle{
{\hat L}^\hbar({\hat S},z)=\sum\limits_{\al} T_\al {\hat S}_\al\vf_\al(z,\om_\al+\frac{\hbar}{N})\,,
 }
 \end{array}
 \eq
  \beq\label{w131}
  \begin{array}{c}
    \displaystyle{
\res\limits_{z=0}{\hat L}({\hat S},z)={\hat S}=\sum\limits_\al T_\al {\hat S}_\al\,,
 }
 \end{array}
 \eq
 where the sum is over $\al\in\mZ_N\times \mZ_N$, the basis matrices\footnote{Some more properties of the
 basis $T_\al$ are briefly reviewed in the Appendix of \cite{ZZ}.} $T_\al$ are given in (\ref{a971}), and
 ${\hat S}_\al$ are the components of the matrix ${\hat S}$ in the basis $T_\al$. In the fundamental representation
 ${\hat S}_\al=(1/N)T_{-\al}$, and in this way one restores the Baxter-Belavin $R$-matrix from (\ref{w13}).

 As a function of the spectral parameter $z$ the Lax operator (\ref{w13}) has the following quasi-periodic behaviour
 on the lattice of periods of the elliptic curve $\mC/(\mZ+\tau\mZ)$:
  \beq\label{w14}
  \begin{array}{c}
    \displaystyle{
{\hat L}({\hat S},z+1)=Q^{-1}{\hat L}({\hat S},z)Q\,,
 }
 \\ \ \\
    \displaystyle{
{\hat L}({\hat S},z+\tau)=\exp(-\frac{2\pi\imath \hbar}{N})\,\Lambda^{-1}{\hat L}({\hat S},z)\Lambda\,,
 }
 \end{array}
 \eq
 where $Q$ and $\Lambda$ are the matrices (\ref{a972}). The properties (\ref{w14}) are derived from
 (\ref{percond}) and (\ref{Tcond0}). The latter yields $Q^{-1}T_\al Q=\exp(2\pi\imath\al_2/N)T_\al$
 and $\Lambda^{-1}T_\al \Lambda=\exp(-2\pi\imath\al_1/N)T_\al$.

 The $L$-operator (\ref{w13}) satisfies the exchange relations (\ref{w08}) identically in $z,w$, thus providing the
 quantum ${\rm GL}_N$ Sklyanin algebra for the set of generators ${\hat S}_\al$, $\al\in\mZ_N\times \mZ_N$.
 The matrix (\ref{w13}) is fixed by the quasi-periodic boundary conditions (\ref{w13}) together with
 fixation of the residue (\ref{w131}) at the single simple pole $z=0$.

\paragraph{Gaudin model.} The Gaudin model \cite{Gaudin} appears from the spin chain
(\ref{w04}) as the limiting case when $\hbar\rightarrow 0$. The ''Planck constant'' $\hbar$ is just a parameter of the model (\ref{w04}), so that the Gaudin model is also quantum, and the only reason to call $\hbar$ the Planck constant is the classical limit expansion %
Plugging (\ref{serRz}) into (\ref{w04}) one obtains the following Gaudin Hamiltonians in the first non-trivial order (in the order $\hbar^{2-n}$):
  \beq\label{w15}
  \begin{array}{c}
    \displaystyle{
 {\hat H}^{\rm G}_i=\sum\limits_{k: k\neq i}^n r_{ik}(z_i-z_k)\,.
 }
 \end{array}
 \eq
 The commutativity of these Hamiltonians follows from the classical Yang-Baxter equation
  In the limit
$\hbar\rightarrow 0$ the Sklyanin algebra (based on (\ref{w08})) turns into the Lie algebra relations
$[{\hat S}^i_1,{\hat S}^j_2]=\delta^{ij}[{\hat S}^i_1,P_{12}]$, and the generators ${\hat S}^i_{(0,0)}$ (the scalar component of the matrix ${\hat S}$) become the Casimirs. Similarly to calculation of the Hamiltonians (\ref{w15}) one can easily obtain the Lax operator for Gaudin model as the first non-trivial term in the expansion in $\hbar$ of the monodromy matrix ${\hat T}(z)$ (\ref{w06}). This yields the Lax operator
  \beq\label{w16}
  \begin{array}{c}
    \displaystyle{
 {\hat l}^{\rm G}(z)=\sum\limits_{k=1}^n \tr_2\Big(r_{12}(z-z_k){\hat S}^k_2\Big)\,.
 }
 \end{array}
 \eq
 Each term in this sum has simple pole at $z=z_k$ with the residue equal to ${\hat S}^k$.
 That is, in the Gaudin limit the multiplicative form of the monodromy matrix (\ref{w05}) turns into
 the additive form of the Lax operator (\ref{w16}), and the quadratic (Sklyanin's) Poisson structure
 turns into the linear Poisson-Lie brackets.

 \subsection{Spin chain as relativistic Gaudin model}

 In this subsection we explain what we mean by the term relativistic Gaudin model. It is just an additive form of the monodromy matrix of spin chain.

 \paragraph{Additive form of spin chain.} Let us represent the monodromy matrix of the spin chain (\ref{w05})
 in the additive form similarly to the Lax operator of the Gaudin model. We begin with the elliptic case.
 The monodromy matrix ${\hat T}(z)$ is an operator valued $N\times N$ matrix. As a function of $z$ it has $n$ simple poles
 at $z=z_i$, $i=1,...,n$. The quasi-periodic behaviour follows from (\ref{w14}):
  \beq\label{w17}
  \begin{array}{c}
    \displaystyle{
{\hat T}(z+1)=Q^{-1}{\hat T}(z)Q\,,
 }
 \\ \ \\
    \displaystyle{
{\hat T}(z+\tau)=\exp(-\frac{2\pi\imath n\hbar}{N})\,\Lambda^{-1}{\hat T}(z)\Lambda\,,
 }
 \end{array}
 \eq
 so that these properties are the same as in (\ref{w14}) but with $\hbar$ being replaced by $n\hbar$.
Therefore, ${\hat T}(z)$ acquires the form:
  \beq\label{w18}
  \begin{array}{c}
    \displaystyle{
{\hat T}(z)=\sum\limits_{k=1}^n {\hat L}^{n\hbar}({\hat\mS}^k,z-z_k)
=\sum\limits_{k=1}^n\sum\limits_{\al} T_\al {\hat \mS}^k_\al\vf_\al(z-z_k,\om_\al+\frac{n\hbar}{N})\,,
 }
 \end{array}
 \eq
 where ${\hat\mS}^k$ are residues of ${\hat T}(z)$ at the poles $z_k$. Namely,
  \beq\label{w19}
  \begin{array}{c}
    \displaystyle{
{\hat\mS}^i=\res\limits_{z=z_i}{\hat T}(z)=
}
\\ \ \\
    \displaystyle{
={\hat L}^\hbar({\hat S}^1,z_i-z_1)\ldots {\hat L}^\hbar({\hat S}^{i-1},z_i-z_{i-1})\cdot
 {\hat S}^{i}\cdot {\hat L}^\hbar({\hat S}^{i+1},z_i-z_{i+1})\ldots {\hat L}^\hbar({\hat S}^{n},z_i-z_{n})\,.
 }
 \end{array}
 \eq
 In this way we express the generators ${\hat\mS}^i_{ab}$,  $i=1,...,n$, $a,b=1,...,N$ in terms of
 the generators of the Sklyanin algebras (${\hat S}^k_{cd}$).
 It is important to mention that the commutation relations between operators ${\hat\mS}^i$ are non-trivial.
 Initially, we had $n$ copies of the Sklyanin algebra, where the operators related to different sites commute, i.e.
 $[{\hat S}^i_{ab},{\hat S}^j_{cd}]=0$  for any $a,b,c,d$ and $i\neq j$. Equivalently, $[{\hat S}_1^i,{\hat S}_2^j]=0$. But it is not true for ${\hat\mS}^i_{ab}$: $[{\hat \mS}_1^i,{\hat \mS}_2^j]\neq 0$. The commutation relations for ${\hat\mS}^i_{ab}$ can be derived from RTT relations (\ref{w09}) by substitution (\ref{w18}).
 These relations can be found in \cite{CLOZ}.

 Notice also that
  \beq\label{w191}
  \begin{array}{c}
    \displaystyle{
\tr({\hat\mS}^i)=\res\limits_{z=z_i}\,\tr\Big({\hat T}(z)\Big)\stackrel{(\ref{w112})}{=}{\hat H}_i\,.
}
 \end{array}
 \eq

\paragraph{Changing the Planck constant parameter.}
Finally, due to (\ref{w07}) from (\ref{w18}) we conclude:
  \beq\label{w201}
  \begin{array}{c}
    \displaystyle{
{\hat T}(z)=\sum\limits_{k=1}^n \tr_2\Big(R^{n\hbar}_{12}(z-z_k){\hat\mS}^k_2\Big)\,.
 }
 \end{array}
 \eq
 This form of ${\hat T}(z)$ is similar to the Lax operator (\ref{w16}) of the Gaudin model. Moreover, one can achieve
 exact matching in the following way. In fact, the constant $n\hbar$ in the $R$-matrix in (\ref{w201}) can be made
 different or even more removed at all. Consider for simplicity $n=1$ case:
 \beq\label{w202}
 \begin{array}{c}
  \displaystyle{
  L^\eta({\hat\mS},z)=\sum\limits_{a} T_a { \hat\mS}_a\vf_a(z,\om_a+\eta)\,,
 }
 \end{array}
 \eq
 Using relation
 \beq\label{w203}
 \begin{array}{c}
  \displaystyle{
  \frac{\vf_a(z-\eta,\om_a+\eta)}{\phi(z-\eta,\eta)}=\frac{\vf_a(z,\om_a)}{\vf_a(\eta,\om_a)},
 }
 \end{array}
 \eq
 it is easy to see that
 \beq\label{w204}
 \begin{array}{c}
  \displaystyle{
   L^\eta(\hat\mS,z)=\phi(z,\eta)\, L^0({\bf S},z+\eta)\,,
 }
 \end{array}
 \eq
 where
 \beq\label{w205}
 \begin{array}{c}
  \displaystyle{
  L^0({\bf S},z)=1_N { \bf S}_0+\sum\limits_{a\neq 0} T_a { \bf S}_a\vf_a(z,\om_a)
 }
 \end{array}
 \eq
and
 \beq\label{w206}
 \begin{array}{c}
  \displaystyle{
  {\hat \mS}=L^0({ \bf S},\eta)\,.
 }
 \end{array}
 \eq
 The latter means explicit change of variables:
 \beq\label{w2061}
 \begin{array}{c}
  \displaystyle{
  {\hat \mS}_0={ \bf S}_0\,,\qquad {\hat \mS}_\al={ \bf S}_\al\vf_\al(\eta,\om_\al)\,,\quad\hbox{for}\ \al\neq 0\,.
 }
 \end{array}
 \eq
 Similar procedure can be performed in the multi-pole case. Then, by redefining the operators $\mS^i$,
 the monodromy matrix (\ref{w201}) takes the form\footnote{An additional pole may arise in the described above procedure, so that the number of simple poles $n'$ may be equal to $n+1$. In (\ref{w207}) an additional generator ${ \bf S}_0$ appears. But it assumed that $\sum\limits_{k=1}^{n'}{\bf S}^k_0=0$ in order to make
  the expression in the r.h.s. of (\ref{w201}) quasi-periodic with respect to $z\rightarrow z+\tau$, that is the total number of independent generators in the scalar component remains $n$.}
  \beq\label{w207}
  \begin{array}{c}
    \displaystyle{
{\hat T}(z-\eta)=f(z)\Big(1_N{ \bf S}_0+\sum\limits_{k=1}^{n'} \tr_2\Big(r_{12}(z-z_k){\bf S}^k_2\Big)\Big)
 }
 \end{array}
 \eq
 with some function $f(z)$. Details are given in \cite{CLOZ}.

 To summarize, the monodromy matrix (\ref{w05}) can be represented in the Gaudin like form (\ref{w16})
 after some set of redefinitions. While in (\ref{w16}) the operators ${\hat S}^i$ are generators of ($i$-th copy of) Lie algebra, in  (\ref{w207}) we have the operators ${\bf S}^i$ originated from $n$ copies of the Sklyanin algebra.

 Another important remark is that the above mentioned trick allows to change the Planck constant parameter entering (\ref{w201}). Indeed, one can remove it in a way described above, and then restore a different parameter. This means that instead of exchange relations (\ref{w08}) one can study more general relations
  \beq\label{w208}
  \begin{array}{c}
    \displaystyle{
 R^\hbar_{12}(z-w){\hat L}^\eta_1({\hat S}^{i},z){\hat L}^\eta_2({\hat S}^{i},w)=
{\hat L}^\eta_2({\hat S}^{i},w){\hat L}^\eta_1({\hat S}^{i},z)R^\hbar_{12}(z-w)
 }
 \end{array}
 \eq
 with two parameters $\hbar$ and $\eta$. It can be shown that these relations are indeed fulfilled for the elliptic Lax operator (\ref{w13}) and the elliptic $R$-matrix (\ref{BB}) in a sense that (\ref{w208}) is equivalent to a set of quadratic algebra relations identically in spectral parameters $z$ and $w$. The quadratic relations explicitly depend on two parameters. But one of them can be removed by the above mentioned redefinitions. Then we are left with a single parameter as it should be in the Sklyanin algebras. However, it is sometimes useful to keep both parameters. We will use this possibility below when studying the classical limit.

\paragraph{Additive form for the fundamental representation.}
Technically, the additive representation is based on the identity
  \beq\label{w20}
  \begin{array}{c}
    \displaystyle{
\prod_{i=1}^n \phi(x_i,y_i)=\sum_{i=1}^n \phi\biggl(x_i,\sum_{m=1}^n y_m\biggr )\prod_{j\neq i}^n \phi(x_j-x_i,y_j)\,,
 }
 \end{array}
 \eq
which is the $n$-th order generalization of the addition formula (\ref{Fay}). Indeed, by definition (\ref{w05})
any matrix element of ${\hat T}(z)$ is a sum of terms, which dependence on $z$ has the form $$\vf_{\alpha_1}(z-z_1,\om_{\alpha_1}+\hbar/N)\ldots\vf_{\alpha_n}(z-z_n,\om_{\alpha_n}+\hbar/N)$$
for some $\al_1,...,\al_n$. Using (\ref{w20}) for $x_i=z-z_i$ one gets (\ref{w18})-(\ref{w19}).

For $R$-matrices satisfying the associative Yang-Baxter equation (\ref{AYB}) there is an $R$-matrix analogue
of the $n$-th order formula (\ref{w20})\footnote{See Section 4 in \cite{MZ}. Similar formula was proved in \cite{Z18}.}:
  \beq\label{w21}
  \begin{array}{l}
  \displaystyle{
R_{0,1}^{y_1}(x_1)R_{0,2}^{y_2}(x_2)\dots R_{0,n}^{y_n}(x_n)=
 }
 \\ \ \\
   \displaystyle{
 =R_{0,n}^{Y}(x_n)\cdot R^{y_1}_{n,1}(x_1-x_n)R^{y_2}_{n,2}(x_2-x_{n})\dots R^{y_{n-1}}_{n,n-1}(x_{n-1}-x_n)+
 }
  \\ \ \\
   \displaystyle{
 +R_{n-1,n}^{y_n}(x_n-x_{n-1})\cdot R_{0,n-1}^Y(x_{n-1})\cdot R^{y_1}_{n-1,1}(x_1-x_{n-1})
 \dots R^{y_{n-2}}_{n-1,n-2}(x_{n-2}-x_{n-1})+
 }
  \\ \ \\
   \displaystyle{
 +R_{n-2,n-1}^{y_{n-1}}(x_{n-1}-x_{n-2})R_{n-2,n}^{y_{n}}(x_{n}-x_{n-2})\cdot
 R_{0,n-2}^Y(x_{n-2})\cdot
  }
  \\ \ \\
   \displaystyle{
 \hspace{60mm}\cdot R^{y_1}_{n-2,1}(x_1-x_{n-2})
 \dots R^{y_{n-3}}_{n-2,n-3}(x_{n-3}-x_{n-2})+
 }
   \\ \ \\
   \displaystyle{
 \vdots
 }
    \\ \ \\
   \displaystyle{
 +R_{1,2}^{y_2}(x_2-x_1)R_{1,3}^{y_3}(x_3-x_1)\dots R_{1,n}^{y_n}(x_{n}-x_{1})\cdot R_{0,1}^Y(x_1)
 \,,
 }
 \end{array}
 \eq
where $Y=\sum\limits_{m=1}^n
y_m$. When $n=2$ it is the equation (\ref{AYB}). In the scalar case ($N=1$) the above identity (\ref{w21})
turns into (\ref{w20}) since $R$-matrices in $N=1$ case become $\phi$-functions.
Plugging $y_1=...=y_n=\hbar$ (so that $Y=n\hbar$) and $x_i=z-z_i$ into (\ref{w21}) one gets
the following additive formula for the monodromy matrix (\ref{w04}):

  \beq\label{w22}
  \begin{array}{l}
  \displaystyle{
{\hat T}(z)=R_{0,1}^{\hbar}(z-z_1)R_{0,2}^{\hbar}(z-z_2)\dots R_{0,n}^{\hbar}(z-z_n)=
 }
 \\ \ \\
   \displaystyle{
 =R_{0,n}^{n\hbar}(z-z_n)\cdot R^{\hbar}_{n,1}(z_n-z_1)R^{\hbar}_{n,2}(z_n-z_2)\dots R^{\hbar}_{n,n-1}(z_{n}-z_{n-1})+
 }
  \\ \ \\
   \displaystyle{
 +R_{n-1,n}^{\hbar}(z_{n-1}-z_{n})\cdot R_{0,n-1}^{n\hbar}(z-z_{n-1})\cdot R^{\hbar}_{n-1,1}(z_{n-1}-z_{1})
 \dots R^{\hbar}_{n-1,n-2}(z_{n-1}-z_{n-2})+
 }
  \\ \ \\
   \displaystyle{
 +R_{n-2,n-1}^{\hbar}(z_{n-2}-z_{n-1})R_{n-2,n}^{\hbar}(z_{n-2}-z_{n})\cdot
 R_{0,n-2}^{n\hbar}(z-z_{n-2})\cdot
  }
  \\ \ \\
   \displaystyle{
 \hspace{60mm}\cdot R^{\hbar}_{n-2,1}(z_{n-2}-z_1)
 \dots R^{\hbar}_{n-2,n-3}(z_{n-2}-z_{n-3})+
 }
   \\ \ \\
   \displaystyle{
 \vdots
 }
    \\ \ \\
   \displaystyle{
 +R_{1,2}^{\hbar}(z_1-z_2)R_{1,3}^{\hbar}(z_1-z_3)\dots R_{1,n}^{\hbar}(z_1-z_n)\cdot R_{0,1}^{n\hbar}(z-z_1)
 \,.
 }
 \end{array}
 \eq
By taking trace over zero tensor component and evaluating residues at $z=z_i$ one easily reproduces (\ref{w113}).

\subsection{Classical models}

  \paragraph{Classical Sklyanin algebra and relativistic integrable tops.}
The model defined in (\ref{w05}) is a quantum version of the model 2 from the Scheme 2. Its classical version was proposed in \cite{Skl1}, see also \cite{Skl2,FTbook}.
 Main idea is very similar to the one described above in the quantum case. In classical mechanics we deal with the Lax matrix of the form
  \beq\label{w23}
  \begin{array}{c}
    \displaystyle{
 L(S,z)=1_N S_0+\sum\limits_{\al\neq 0} T_\al S_\al \vf_\al(z,\om_\al)\,,
 }
 \end{array}
 \eq
 which is similar (\ref{w207}), but here $S=\sum\limits_{\al}T_\al S_\al\in\Mat$ is a matrix of $N^2$ dynamical variables $S_\al=S_{\al_1,\al_2}$ (coordinates on the phase space). The Poisson structure is generated by the quadratic $r$-matrix structure
  \beq\label{w24}
  \begin{array}{c}
    \displaystyle{
 \{L_1(S,z),L_2(S,w)\}=[r_{12}(z-w),L_1(S,z)L_2(S,w)]\,,
 }
 \end{array}
 \eq
 where
  \beq\label{w25}
  \begin{array}{c}
    \displaystyle{
 \{L_1(S,z),L_2(S,w)\}=\sum\limits_{i,j,k,l=1}^N \{L_{ij}(S,z),L_{kl}(S,w)\}E_{ij}\otimes E_{kl}=
 \sum\limits_{\al,\be\in\,\mZ_N^{\times 2}} \{L_{\al}(S,z),L_{\be}(S,w)\} T_\al\otimes T_\be
 }
 \end{array}
 \eq
and $r_{12}(z-w)$ is the classical elliptic $r$-matrix. It can be shown that (\ref{w24}) is identically fulfilled
in $z,w$ and provides the set of Poisson brackets $\{S_\al,S_\be\}$, which is called the classical Sklyanin algebra.
The underlying integrable system is the relativistic elliptic top. It is in the box 5 on the Scheme 2. Let us notice
that the first flow generated by the Hamiltonian $H=S_0$ provides equations of motion, which have precisely the same form as those in the non-relativistic case (for the model 5 on the Scheme 1) generated by the Hamiltonian
$H=(-1/2)\sum_{\al\neq 0} S_\al S_{-\al}\wp(\om_\al)$ and the linear Poisson-Lie brackets. This phenomenon
reflects existence of bi-Hamiltonian structure. See details in \cite{KLO}.

Following \cite{LOZ8} we slightly change the above definitions (\ref{w23})-(\ref{w24}). Namely, we consider the Lax matrix with explicit dependence on the parameter $\eta$:
  \beq\label{w26}
  \begin{array}{c}
    \displaystyle{
 L^\eta(S,z)=\sum\limits_{\al} T_\al S_\al \vf_\al(z,\om_\al+\frac{\eta}{N})
 }
 \end{array}
 \eq
 or
  \beq\label{w261}
  \begin{array}{c}
    \displaystyle{
 L^\eta(S,z)=\tr_2\Big(R_{12}^\eta(z)S_2\Big)\,,
 }
 \end{array}
 \eq
 which is obtained from the elliptic quantum $L$-operator ${\hat L}^\eta({\hat S},z)$ by replacing ${\hat S}$ with $S$. The relation between descriptions in terms of Lax matrices  (\ref{w23}) and (\ref{w26}) is the same as in (\ref{w206}). So that in (\ref{w261}) we added by hands explicit dependence on the additional parameter $\eta$. The corresponding Sklyanin algebra is now generated by
 \beq\label{w27}
 \begin{array}{c}
  \displaystyle{
 \{L_1^{\eta}(z,S),L_2^{\eta}(w,S)\}=\frac{1}{c}\,[L_1^{\eta}(z,S)L_2^{\eta}(S,w),r_{12}(z-w)]\,,
 }
 \end{array}
 \eq
 where $c$ is another constant parameter.
 It is straightforwardly follows from the quantum exchange relations (\ref{w208}) in the limit $\hbar\rightarrow 0$.
 Namely, one should make a substitution $\hbar\rightarrow (1/c)\hbar$ and then consider the classical limit (\ref{serRz}) with the standard definition
 \beq\label{w28}
 \begin{array}{c}
  \displaystyle{
 \{L_1^{\eta}(z,S),L_2^{\eta}(w,S)\}=\lim\limits_{\hbar\rightarrow 0}\frac{ {\hat L}_1^{\eta}(z,{\hat S}) {\hat L}_2^{\eta}(w,{\hat S})-{\hat L}_2^{\eta}(w,{\hat S}) {\hat L}_1^{\eta}(z,{\hat S}) }{\hbar}
 }
 \end{array}
 \eq
 In deriving this relation from (\ref{w208}) we used independence of parameters $\eta$ and $\hbar$. Also, compared to (\ref{w24}), we put a factor $(-1/c)$ in the r.h.s. This factor is just for convenience of describing relation to Ruijsenaars-Schneider model (see below).

 Direct computations show that (\ref{w28}) is equivalent to the following set of Poisson brackets in the classical Sklyanin algebra:
  \beq\label{w29}
 \begin{array}{c}
   \displaystyle{
 \{S_1,S_2\}
 =\frac{1}{c}\,[S_1S_2,r_{12}^{(0)}]+\frac{1}{c}\,[\tr_3(R_{13}^{\eta,(0)}S_3)S_2,P_{12}]\,,\
 \tr_3(R_{13}^{\eta,(0)}S_3)=\sum\limits_\al T_\al S_\al \Big(E_1(\om_\al\!+\!\frac\eta{N})\!+\!\frac{2\pi\imath\al_2}{N}\Big)\,,
 }
 \end{array}
  \eq
where $r_{12}^{(0)}$ and $R_{13}^{\eta,(0)}$ come from the expansions (\ref{CYB2}) and (\ref{serRx}).

The relativistic top (model 5 on the Scheme 2) is defined as follows. The Poisson brackets (\ref{w29}) together
with the Hamiltonian
 \beq\label{w30}
 \begin{array}{c}
  \displaystyle{
 H^{top}=cN S_0=c\,\tr\, S=c\,\frac{\tr L^\eta(S,z)}{\phi(z,\eta/N)}\,.
 }
 \end{array}
 \eq
generate dynamics given by the following equations of motion:
 \beq\label{w31}
 \begin{array}{c}
  \displaystyle{
\dot S=[S,J^\eta(S)]\,.
 }
 \end{array}
 \eq
 They have the form of multi-dimensional Euler-Arnold top. The linear operator $J^\eta$ plays the role
 of the inverse tensor of inertia (in principal axes). It has the form:
 \beq\label{w32}
 \begin{array}{c}
  \displaystyle{
 J^\eta(S)=1_N S_0 E_1(\eta/N)+\sum\limits_{\al\in\mZ_N^{\times 2}, \al\neq
0}T_\al S_\al J_\al^\eta\,,\quad
 J_\al^\eta=E_1(\eta/N+\om_\al)-E_1(\om_\al)\,.
 }
 \end{array}
 \eq
The equations (\ref{w31}) are represented in the Lax form
 \beq\label{w34}
 \begin{array}{c}
  \displaystyle{
 {\dot L}^{\eta}(S,z)=\{H^{top},L^{\eta}(S,z)\}=[L^{\eta}(S,z),M(S,z)]
 }
 \end{array}
 \eq
with the $M$-matrix
 \beq\label{w35}
 \begin{array}{c}
  \displaystyle{
 M(S,z)=-\tr_2\Big(r_{12}(z)S_2\Big)\,.
 }
 \end{array}
 \eq
In the elliptic case the above statement is verified directly using identities from Appendix A.
At the same time the construction of the relativistic top can be generalized to any solution
of the associative Yang-Baxter equation (\ref{AYB}). Indeed, the definition of the Lax pair (\ref{w261})
and (\ref{w35}) does not use explicit form of the underlying $R$-matrix. The calculation providing the proof of the Lax equations can be performed using $R$-matrix identities coming from (\ref{AYB}). This type identities are collected in the Appendix B. In this case the expression $J^\eta(S)$ acquires the following form:
  \beq\label{w36}
  \begin{array}{c}
      \displaystyle{
  J^\eta(S)=\tr_2\Big(J^\eta_{12}S_2\Big)\,,\qquad   J^\eta_{12}=R^{\eta,(0)}_{12}-r^{(0)}_{12}\,,
}
  \end{array}
  \eq
where $R^{\eta,(0)}_{12}$ is the coefficient of expansion (\ref{serRx}), and $r^{(0)}_{12}$ is the coefficient of expansion (\ref{CYB2}).
The description of the (generalized) relativistic top in terms of $R$-matrices  was proposed in \cite{LOZ8} and then proved
 in \cite{LOZ2} and \cite{KZ19}. Finally, let us remark that the described above integrable top can be viewed as the spin chain on a single site.

 \paragraph{Classical spin chains.}
Next, we proceed to the classical spin chain on $n$ sites by introducing the classical monodromy matrix
  \beq\label{w37}
  \begin{array}{c}
    \displaystyle{
 T(z)=L^\eta(S^1,z-z_1)\ldots L^\eta(S^n,z-z_n)\,,
 }
 \end{array}
 \eq
where $S^1,...,S^n\in\Mat$ are $n$ matrices of size $N\times N$ of dynamical variables. Due to (\ref{w261}) it can be also represented in the form:
  \beq\label{w38}
  \begin{array}{c}
    \displaystyle{
 T_0(z)=\tr_{12...n}\Big({\hat T}(z)S^1_1S^2_2...S^n_n\Big)\,,
 }
 \end{array}
 \eq
 where (similarly to notations $S_1=S\otimes 1_N$ and $S_2=1_N\otimes S$) $S^i_i$ means the $S^i$ matrix in the $i$-th tensor component\footnote{More precisely, in (\ref{w38}) we assume $S_i^i$ be the $S^i$ matrix in the $i+1$-th tensor component. It is because $S^i$ is an element of $\Mat^{\otimes(n+1)}$ since the first tensor component has number 0 (it is the matrix space of $T(z)$).}, while ${\hat T}(z)$ is the quantum monodromy matrix in the fundamental representation (\ref{w04}).

The Poisson structure for the spin chain is given by $n$ copies of the Sklyanin algebra generated by $n$ copies of
the quadratic $r$-matrix structure
 \beq\label{w39}
 \begin{array}{c}
  \displaystyle{
 \{L_1^{\eta}(S^i,z),L_2^{\eta}(S^j,w)\}=\frac{\delta^{ij}}{c}\,[L_1^{\eta}(S^i,z)L_2^{\eta}(S^i,w),r_{12}(z-w)]\,,
 }
 \end{array}
 \eq
 so that any Poisson brackets between variables from different sites vanish. The monodromy matrix satisfies the same relations
 \beq\label{w40}
 \begin{array}{c}
  \displaystyle{
 \{T_1(z),T_2(w)\}=\frac{1}{c}\,[T_1(z)T_2(w),r_{12}(z-w)]\,.
 }
 \end{array}
 \eq
Therefore, the classical transfer matrix
 \beq\label{w401}
 \begin{array}{c}
  \displaystyle{
 t(z)=\tr T(z)
 }
 \end{array}
 \eq
 is a generating function of the classical Hamiltonians commuting with respect to the Poisson structure given by a direct sum of $n$ Sklyanin algebras.

 Further description is parallel to the quantum case. One can represent the monodromy matrix (\ref{w37}) in the form
  \beq\label{w41}
  \begin{array}{c}
    \displaystyle{
{ T}(z)=\sum\limits_{k=1}^n { L}^{n\eta}({\mS}^k,z-z_k)
=\sum\limits_{k=1}^n\sum\limits_{\al} T_\al { \mS}^k_\al\vf_\al(z-z_k,\om_\al+\frac{n\eta}{N})\,,
 }
 \end{array}
 \eq
 where ${\mS}^k$ are again residues of ${ T}(z)$ at the poles $z_k$, i.e.
  \beq\label{w42}
  \begin{array}{c}
    \displaystyle{
{\mS}^i=\res\limits_{z=z_i}{ T}(z)
={ L}^\eta({ S}^1,z_i-z_1)\ldots { L}^\eta({ S}^{i-1},z_i-z_{i-1})\cdot
 { S}^{i}\cdot { L}^\eta({ S}^{i+1},z_i-z_{i+1})\ldots { L}^\eta({ S}^{n},z_i-z_{n})\,.
 }
 \end{array}
 \eq
The non-local Hamiltonians are the classical analogues of (\ref{w112}):
  \beq\label{w421}
  \begin{array}{c}
      \displaystyle{
 {H}_i=\res\limits_{z=z_i}\,\tr\Big( T(z)\Big)=\tr(\mS^i)=
 }
  \\ \ \\
    \displaystyle{
=\tr \Big( { L}^\eta({ S}^1,z_i-z_1)\ldots { L}^\eta({ S}^{i-1},z_i-z_{i-1})\cdot
 { S}^{i}\cdot { L}^\eta({ S}^{i+1},z_i-z_{i+1})\ldots { L}^\eta({ S}^{n},z_i-z_{n}) \Big)\,.
 }
 \end{array}
 \eq
 In this way we come to the additive form of the monodromy matrix. It can be view as relativistic Gaudin model by the following reason. Using the change of variables of type (\ref{w2061}) one can (similarly to the quantum case (\ref{w207})) represent $T(z)$ in the form (see the footnote for (\ref{w207})):
  \beq\label{w43}
  \begin{array}{c}
    \displaystyle{
{ T}(z-\eta)=f(z)\Big(1_N{ \ti S}_0+\sum\limits_{k=1}^n \tr_2\Big(r_{12}(z-z_k){\ti S}^k_2\Big)\Big)
 }
 \end{array}
 \eq
 or
  \beq\label{w44}
  \begin{array}{c}
    \displaystyle{
 { T}(z-\eta)=f(z)\Big(1_N{ \ti S}_0+\Big(\sum\limits_{k=1}^n 1_N{ \ti S}^k_0 E_1(z-z_k) +\sum\limits_{\ga\neq 0}
  T_\ga{\ti S}^k_\ga \vf_\ga(z-z_k,\om_\ga)\Big)\Big)\,.
 }
 \end{array}
 \eq
 It is the form of the Lax matrix for the classical Gaudin model \cite{STS}, which we considered in
 our previous paper \cite{TZ}. In that model the Poisson structure
 is given by the linear Poisson-Lie brackets, while in (\ref{w43}) we deal with some quadratic Poisson algebra coming from $n$ copies of Sklyanin algebra via (\ref{w42}). Explicit formulae for the brackets can be found in \cite{CLOZ}. The term relativistic Gaudin model implies the multi-pole (and multispin) structure together
 with quadratic Poisson brackets.

 Let us remark that the above description naturally arises in the Hitchin approach to integrable systems, where
 the Lax matrices are considered as sections of a certain bundles over curves. The relativistic
 generalization of Hitchin systems was studied in \cite{BDOZ,CLOZ} and \cite{EPestun}.

 \paragraph{Gauge equivalence between RS model and relativistic top.}
Here, following \cite{Hasegawa} (see also \cite{Chen,KZ19,VZ,ZZ}) we briefly describe
the change of variables between the models 8 and 9 on the Scheme 2.

The $N$-body Ruijsenaars-Schneider model \cite{Ruij} is given by the following Lax matrix of size $N\times N$:
  \beq\label{w45}
  \begin{array}{l}
  \displaystyle{
 L^{\rm RS}_{ij}(z)=\phi(z,q_{ij}+\eta)\,b_j\,,\ i,j=1,\ldots ,N\,,
 }
 \end{array}
 \eq
 where
  \beq\label{w46}
  \begin{array}{l}
  \displaystyle{
 b_j=\prod_{k:k\neq j}^N\frac{\vth(q_{j}-q_k-\eta)}{\vth(q_{j}-q_k)}\,
 e^{p_j/c}\,,\quad c={\rm const}\in\mC\,.
 }
 \end{array}
 \eq
The Hamiltonian
  \beq\label{w47}
  \begin{array}{l}
  \displaystyle{
 H^{\rm RS}=c\frac{\tr L^{\rm RS}(z)}{\phi(z,\eta)}=c\sum\limits_{j=1}^N b_j(p,q)
 }
 \end{array}
 \eq
 with the canonical Poisson brackets
$
\{p_i,q_j\}=\delta_{ij}$ (and $\{p_i,p_j\}=\{q_i,q_j\}=0
$)
 generates equations of motion
  \beq\label{w48}
  \begin{array}{l}
  \displaystyle{
 {\ddot q}_i=\sum\limits_{k:k\neq i}^N{\dot q}_i{\dot q}_k
 (2E_1(q_{ik})-E_1(q_{ik}+\eta)-E_1(q_{ik}-\eta))\,,\quad i=1, \ldots N\,.
 }
 \end{array}
 \eq
 Introduce the elliptic intertwining matrix \cite{Baxter3}:
  \beq\label{w49}
  \begin{array}{l}
  \displaystyle{
 g_{ij}(z,q)=
 \vth\left[  \begin{array}{c}
 \frac12-\frac{i}{N} \\ \frac N2
 \end{array} \right] \left(z-Nq_j+\sum\limits_{m=1}^N
 q_m\left.\right|N\tau\right)\,\frac{1}{\prod\limits_{k:k\neq j}^N\vth(q_j-q_k)}\,,
 }
 \end{array}
 \eq
where
the theta-functions with characteristics appear:
\beq\label{w50}
 \begin{array}{c}
  \displaystyle{
\theta{\left[\begin{array}{c}
a\\
b
\end{array}
\right]}(z|\, \tau ) =\sum_{j\in \mZ}
\exp\left(2\pi\imath(j+a)^2\frac\tau2+2\pi\imath(j+a)(z+b)\right)\,.
}
 \end{array}
 \eq
 This matrix was used to describe the IRF-Vertex correspondence in 2d integrable lattice models.

The announced relation between models 8 and 9 consists of two steps. The first one is that the Lax matrix
(\ref{w45}) is represented in the factorized form:
  \beq\label{w51}
  \begin{array}{l}
  \displaystyle{
 L^{\rm RS}(z)=\frac{\vth'(0)}{\vth(\eta)}\,
 g^{-1}(z,q)g(z+N\eta,q)\,e^{P/c}\,,\quad P={\rm diag}(p_1,\ldots ,p_N)\,.
 }
 \end{array}
 \eq
The second step is the statement that the gauge transformed matrix $g(z,q)L^{\rm RS}(z)g^{-1}(z,q)$ has the form
of the Lax matrix of relativistic top (\ref{w26}). Namely,
  \beq\label{w52}
  \begin{array}{l}
  \displaystyle{
 L^{N\eta}(S,z)=g(z,q)L^{\rm RS}(z)g^{-1}(z,q)=\frac{\vth'(0)}{\vth(\eta)}\,
 g(z+N\eta,q)\,e^{P/c}g^{-1}(z,q)\,.
 }
 \end{array}
 \eq
 It is a nontrivial exercise to show that the r.h.s. of (\ref{w52}) indeed has the form
  \beq\label{w521}
  \begin{array}{l}
  \displaystyle{
 L^{N\eta}(S,z)=\sum\limits_{\al}S_\al T_\al \vf_\al(z,\om_\al+\eta)
 }
 \end{array}
 \eq
 with some matrix $S$ (see \cite{Hasegawa,Chen,VZ} and the appendix in \cite{ZZ} for details).
The matrix $S$ in this case is special. Its rank equals one, and due to (\ref{w52}) it is a function of the
canonical variables $p_i$, $q_j$. Explicit change of variables can be calculated:
  \beq\label{w53}
  \begin{array}{c}
  \displaystyle{
 S_a(p,q,\eta,c)=\frac{(-1)^{a_1+a_2}}{N}\,e^{\pi\imath a_2\om_a}
 \sum\limits_{m=1}^N e^{p_m/c} e^{2\pi\imath a_2(\eta-{\bar q}_m)}
 \frac{\vth(\eta+\om_a)}{\vth(\eta)}
 \prod\limits_{l:\,l\neq m}^N\frac{\vth(q_m-q_l-\eta-\om_a)}{\vth(q_m-q_l)}\,,
 }
 \end{array}
 \eq
 where $a\in\mZ_N\times\mZ_N$ and ${\bar q}_m=q_m-(1/N)\sum\limits_{k=1}^N q_k$ is the coordinate in the center of masses frame.

\section{Elliptic Lax pairs}
\label{sec3}
\setcounter{equation}{0}

In this Section we consider the most general model 1 from the Scheme 2. Our purpose is to propose
the Lax pair and derive equations of motion. Then we briefly consider some particular cases including
the models 2, 3 and 4.

For all the models from the family III the Lax matrices are of size $NM\times NM$ with a natural block-matrix structure:
 \beq\label{w55}
 \mL(z)= \left.\left(\begin{array}{cccc}
  \mL^{11}(z)  &  \mL^{12}(z) &  \ldots & \mL^{1M}(z)
  \\ \ \\
 \mL^{21}(z)  & \mL^{22}(z)    & \ldots & \mL^{2M}(z)
 \\
 \vdots & \vdots &  \ddots & \vdots \\
 \mL^{M1}(z)  & \mL^{M2}(z)  & \ldots & \mL^{MM}(z)
 \end{array}\right)\quad \right\}
 \
\begin{array}{c}
  \hbox{each column or row}
  \\
  \hbox{contains}\ M\ \hbox{blocks}
  \\
  \hbox{of size}\ N\times N
\end{array}
 \eq
Equivalently,
 \beq\label{w56}
 \begin{array}{c}
  \displaystyle{
  \mL(z)=\sum\limits_{i,j=1}^M E_{ij}\otimes
  \mL^{ij}(z)\in\MatNM\,,\quad  \mL^{ij}(z)\in\Mat\,.
 }
 \end{array}
 \eq
 Inside $N\times N$ blocks (that is inside $N\times N$ matrices $\mL^{ij}(z)$) we use the basis (\ref{a971})
 as we did for relativistic top (\ref{w26}). A similar block-matrix structure is used for (the accompany) $M$-matrix entering the Lax equation $\dot\mL(z)=[\mL(z),\mM(z)]$ and the residues $\mS^c$, $c=1,...,n$ of $\mL(z)$  at simple poles $z_1,...,z_n$, which are the classical spin variables:
 \beq\label{w561}
 \begin{array}{c}
  \displaystyle{
  \mS^a=\sum\limits_{i,j=1}^M E_{ij}\otimes \mS^{ij,a}\,,\quad \mS^a\in{\rm Mat}(NM,\mC)\,,\quad
  \mS^{ij,a}\in\Mat\,.
 }
 \end{array}
 \eq
 Each matrix $\mS^{ij,a}$ has components $\mS^{ij,a}_\ga$, $\ga\in\mZ_N\times \mZ_N$ in the basis $T_\ga$ (\ref{a971}). The zero component (for $T_{0,0}=1_N$) is denoted as either $\mS^{ii,a}_{0,0}$ or just $\mS^{ii,a}_0$.

\subsection{General case}\label{sec31}

The Lax pair for the general model has the block-matrix structure (\ref{w55}) with the $N\times N$ blocks
\begin{equation}\begin{array}{c}
    \displaystyle{ \label{relLaxl} \mL^{ij}(z) =
     \sum_{\gamma\in\mZ_N^{\times 2}} \sum_{a=1}^n T_\gamma \mS^{ij,a}_\gamma \vf_\gamma(z-z_a, \om_\gamma + \frac{q_{ij} + \eta}{N}),
     \quad q_{ij} = q_i - q_j, \quad \om_\gamma = \frac{\gamma_1 + \gamma_2 \tau}{N}}
\end{array}\end{equation}
and
\begin{equation}\begin{array}{c}
\label{relLaxm}
\displaystyle{
  \mM(z)=\sum\limits_{i,j=1}^M E_{ij}\otimes
  \mM^{ij}(z)\in\MatNM\,,\quad  \mM^{ij}(z)\in\Mat\,.
 }
 \\ \ \\
    \displaystyle{ \mM^{ii}(z) = - \sum_{a=1}^n T_0 \mS^{ii,a}_{0,0} E_1(z-z_a)  - \sum_{\gamma \neq 0} \sum_{a=1}^n  T_\gamma \mS^{ii,a}_\gamma \vf_\gamma (z-z_a, \om_\gamma), } \\ \ \\
    \displaystyle{ \mM^{ij}(z) = - \sum_{\gamma} \sum_{a=1}^n T_\gamma \mS^{ij,a}_\gamma \vf_\gamma(z-z_a, \om_\gamma + \frac{q_{ij}}{N}), \quad i\neq j.
}\end{array}\end{equation}
Introduce the following set of linear operators (analogues of the inverse inertia tensor $J^\eta$ (\ref{w31})-(\ref{w32})):
\begin{equation}\begin{array}{c} \displaystyle{
\label{inerten1}
    \widetilde{J}^{\eta, q_{mn}}_a (\mS^{ij,b}) = \sum_{\gamma} \mS^{ij,b}_\gamma T_\gamma \Big( \vf_\gamma (z_{ab}, \om_\gamma + \frac{q_{mn} + \eta}{N}) - \vf_\gamma (z_{ab}, \om_\gamma + \frac{q_{mn}}{N}) \Big)\,,\quad\hbox{for}\ a\neq b\,,
}\end{array}\end{equation}
\begin{equation}\begin{array}{c} \displaystyle{
\label{inerten2}
    \widetilde{J}^{\eta}_a (\mS^{ij,b}) = \sum_{\gamma \neq 0} \mS^{ij,b}_\gamma T_\gamma \Big( \vf_\gamma (z_{ab}, \om_\gamma + \frac{\eta}{N}) - \vf_\gamma (z_{ab}, \om_\gamma) \Big)\,,\quad\hbox{for}\ a\neq b\,,
}\end{array}\end{equation}
\begin{equation}
\label{inerten3}
    J^{\eta, q_{mn}}(\mS^{ij,b}) = \sum_\gamma \mS^{ij,b}_\gamma T_\gamma \Big( E_1(\omega_\gamma + \frac{q_{mn} + \eta}{N}) - E_1(\omega_\gamma + \frac{q_{mn}}{N}) \Big),
\end{equation}
\begin{equation}
\label{inerten4}
    J^\eta(\mS^{ij,b}) = \sum_{\gamma \neq 0} \mS^{ij,b}_\gamma T_\gamma \Big( E_1(\omega_\gamma + \frac{\eta}{N}) - E_1(\omega_\gamma) \Big).
\end{equation}
In the above definitions $i,j,m,n=1,...,M$ and $a,b=1,...,n$. The summation in $\ga$ is over $\mZ_N\times\mZ_N$. If $\ga\neq 0$ then the summation is over $\mZ_N\times\mZ_N\setminus (0,0)$.

Main purpose of the current subsection is to derive equations of motion for the general model. Let us write down the answer for the diagonal and non-diagonal blocks separately.
For non-diagonal blocks $\mS^{ij,a}$ ($i\neq j$) equations of motion have the form:
\begin{equation}\begin{array}{c}
\label{eqij}
    \displaystyle{
    \Dot{ \mS }^{ij,a} = \mS^{ij,a} J^\eta (\mS^{jj, a}) - J^\eta ( \mS^{ii,a} ) \mS^{ij,a} + \sum_{k:k\neq j}^M \mS^{ik,a} J^{\eta, q_{kj}} (\mS^{kj,a}) - \sum_{k:k\neq i}^M J^{\eta, q_{ik}} (\mS^{ik,a}) \mS^{kj,a} +
     }
    \\ \ \\
    \displaystyle{
    + \sum_{b: b\neq a}^n \mS^{ij,a} \Big(\mS^{ii,b}_{0,0} - \mS^{jj,b}_{0,0}\Big) \Big( E_1(z_{ab}) - \phi(z_{ab}, \frac{\eta}{N}) \Big) +
    E_1(\frac{\eta}{N})\mS^{ij,a} \Big(\mS^{ii,a}_{0,0} - \mS^{jj,a}_{0,0}\Big)+
     }
     \\ \ \\
    \displaystyle{
     + \sum_{b: b\neq a}^n \Big(\mS^{ij,a} \widetilde{J}^\eta_a (\mS^{jj,b}) - \widetilde{ J }^{\eta}_a ( \mS^{ii,b} ) \mS^{ij,a} \Big) + \sum_{b: b\neq a}^n \Big( \sum_{k:k\neq j}^M \mS^{ik,a} \widetilde{J}^{\eta, q_{kj}}_a (\mS^{kj, b}) - \sum_{k:k\neq i}^M \widetilde{J}^{\eta, q_{ik}}_a (\mS^{ik, b}) \mS^{kj,a} \Big)\,.
     }
\end{array}\end{equation}
Equations of motion for the diagonal blocks are as follows:
\begin{equation}\begin{array}{c}
\label{eqii}
    \displaystyle{ \Dot{\mS}^{ii,a} =\left[ \mS^{ii,a}, J^\eta  (\mS^{ii,a}) \right] +
     \sum_{k: k\neq i}^M \Big( \mS^{ik,a} J^{\eta, q_{ki}} (\mS^{ki,a}) - J^{\eta, q_{ik}} (\mS^{ik,a}) \mS^{ki,a} \Big) +
     }
     \\ \ \\
     \displaystyle{
      + \sum_{b: b\neq a}^n \left[ \mS^{ii,a}, \widetilde{J}^{\eta}_a (\mS^{ii,b}) \right]
      +\sum_{b: b\neq a}^n \sum_{k: k\neq i}^M \Big( \mS^{ik,a} \widetilde{J}^{\eta, q_{ki}}_a (\mS^{ki,b}) - \widetilde{J}^{\eta, q_{ik}}_a (\mS^{ik,b}) \mS^{ki,a}\Big)\,.
      }
\end{array}\end{equation}

Let us formulate the statement on these equations.

\begin{theorem}\label{th1}
Equations of motion \eqref{eqij} and \eqref{eqii} are equal to the Lax equation with additional term
\begin{equation}\begin{array}{c} \displaystyle{
\label{nonLaxrel}
    \frac{d}{dt} \mL(z) = [ \mL(z), \mM(z) ] + \sum_{i, j=1}^M \sum_{c=1}^n \sum_{\al\in\mZ_N^{\times 2}} \mS^{ij, c}_\al (\mu_i - \mu_j) E_{ij} \otimes T_\al f_\al (z-z_c, \om_\al + \frac{q_{ij} + \eta}{N})
}\end{array}\end{equation}
for the Lax pair \eqref{relLaxl} and \eqref{relLaxm}, where
\begin{equation}
\begin{array}{c}
\displaystyle{ \label{cond}
    \mu_i = \frac{\Dot{q}_i}{N} -  \sum_{a=1}^n \mS^{ii,a}_{0,0}, \quad i=1, \ldots, M
}
\end{array}\end{equation}
\end{theorem}
and the functions $f_\al$ in the additional term in the r.h.s. of (\ref{nonLaxrel}) are given by (\ref{varf}).

\begin{proof}
Consider first the equation (\ref{nonLaxrel}) for the diagonal blocks. For the l.h.s. we have
\begin{equation}
\label{leftii}
    \frac{d}{dt} \mathcal{L}^{ii}(z) =  \sum_{a=1}^n \sum_\gamma T_\gamma \Dot{\mathcal{S}}^{ii,a}_\gamma \varphi_\gamma (z-z_a, \omega_\gamma + \frac{\eta}{N})\,,
\end{equation}
and for the r.h.s.
\begin{equation}\label{w60} \begin{array}{c}
     \displaystyle{\left[ \mathcal{L}, \mathcal{M} \right]^{ii}  = \sum_{a,b=1}^n \sum_{\gamma \neq 0} \sum_\beta ( \kappa_{\gamma, \beta} - \kappa_{\beta, \gamma} ) \mS^{ii,a}_\gamma \mS^{ii,b}_\beta T_{\gamma + \beta} \varphi_\gamma (z-z_a, \omega_\gamma) \varphi_\beta (z-z_b, \omega_\beta + \frac{\eta}{N}) + }
     \\ \ \\
     \displaystyle{ + \sum_{k: k\neq i}^M \sum_{a,b=1}^n \sum_{\gamma, \beta} \kappa_{\gamma, \beta} \mS_\gamma^{ik, a} \mS_\beta^{ki, b} T_{\gamma + \beta} \Big( \varphi_\gamma (z-z_a, \omega_\gamma + \frac{q_{ik}}{N}) \varphi_\beta (z-z_b, \omega_\beta + \frac{q_{ki} + \eta}{N}) - }
     \\ \ \\
    \displaystyle{ - \varphi_\gamma (z-z_a, \omega_\gamma + \frac{q_{ik} + \eta}{N}) \varphi_\beta (z-z_b, \omega_\beta + \frac{q_{ki}}{N} ) \Big).}
\end{array} \end{equation}
The first sum (the upper line) consists of two parts: $a=b$ and $a \neq b$. Consideration of the case $a=b$ coincides
 with the one described in \cite{Reltops}. That gives us the first term  in the equation \eqref{eqii} (also the upper line in the r.h.s.).
  The second term in \eqref{eqii} comes from the second line of (\ref{w60}).
   Consider the case $a\neq b$. By applying the Fay identity (\ref{Fay}) and \eqref{a16}, \eqref{a17} to the first term in the r.h.s. of  (\ref{w60}) we obtain:
\begin{equation} \begin{array}{c}
        \displaystyle{\sum_{a\neq b}^n \sum_\gamma \sum_{\beta \neq 0} \mS_\gamma^{ii, a} \mS_\beta^{ii, b} T_\gamma T_\beta \varphi_{\gamma + \beta}(z - z_a, \omega_{\gamma + \beta} + \frac{\eta}{N}) \Big( \varphi_\beta (z_{ab}, \omega_\beta + \frac{\eta}{N}) - \varphi_\beta(z_{ab}, \omega_\beta) \Big) - }
        \\ \ \\
        \displaystyle{- \sum_{a\neq b}^n \sum_\gamma \sum_{\beta \neq 0} \mathcal{S}_\gamma^{ii, a} \mathcal{S}_\beta^{ii, b} T_\beta T_\gamma \varphi_{\gamma + \beta}(z - z_a, \omega_{\gamma + \beta} + \frac{\eta}{N}) \Big(  \varphi_\beta (z_{ab}, \omega_\beta + \frac{\eta}{N}) - \varphi_\beta (z_{ab}, \omega_\beta)\Big)\,.}
\end{array} \end{equation}
Similarly, for the second term (the second line) in the r.h.s. of (\ref{w60}):
\begin{equation} \begin{array}{c}
    \displaystyle{\sum_{k\neq i}^M\sum_{a\neq b}^n \sum_{\gamma, \beta} \mS^{ik, a}_\gamma \mS^{ki, b}_\beta T_{\gamma} T_{\beta} \varphi_{\gamma + \beta} (z-z_a, \omega_{\gamma +\beta} + \frac{\eta}{N}) \Big( \varphi_\beta (z_{ab}, \omega_\beta + \frac{q_{ki} + \eta}{N}) - \varphi_\beta (z_{ab}, \omega_\beta + \frac{q_{ki}}{N}) \Big) - }\\ \ \\
    \displaystyle{-  \mS^{ik, b}_\beta \mS^{ki, a}_\gamma T_{\beta} T_{\gamma} \Big( \varphi_\beta (z_{ab}, \omega_\beta + \frac{q_{ik} + \eta}{N}) - \varphi_\beta (z_{ab}, \omega_\beta + \frac{q_{ik}}{N}) \Big) \varphi_{\gamma + \beta} (z-z_a, \omega_{\gamma +\beta} + \frac{\eta}{N})}\,.
\end{array} \end{equation}
In this way we obtain the r.h.s. of (\ref{w60}). By comparing it with \eqref{leftii} one gets \eqref{eqii}.

For the non-diagonal blocks of the equation (\ref{nonLaxrel}) we have in the l.h.s.:
\begin{equation}\label{w61}
    \frac{d}{dt} \mathcal{L}^{ij}(z) = \sum_{a=1}^n \sum_\gamma
     \Big(\Dot{\mS}^{ij,a}_\gamma T_\gamma \varphi_\gamma (z-z_a, \omega_\gamma + \frac{q_{ij} + \eta}{N}) + \frac{\Dot{q}_{ij}}{N} \mS^{ij,a}_\gamma T_\gamma f_\gamma (z-z_a, \omega_\gamma + \frac{q_{ij} + \eta}{N})\Big)\,.
\end{equation}
In the r.h.s. of (\ref{nonLaxrel}) the following expression arises:
\begin{equation} \begin{array}{c}\label{w62}
    \displaystyle{
    \sum_{a,b=1}^n \sum_\gamma \mS^{ij,b}_\gamma (\mS^{ii,a}_{0,0} - \mS^{jj,a}_{0,0}) T_\gamma \varphi_\gamma(z - z_b, \omega_\gamma + \frac{ q_{ij} + \eta}{N}) E_1 (z-z_a)
    }
     \\ \ \\
    \displaystyle{+ \sum_{a,b=1}^n \sum_{\substack{\gamma \neq 0 \\ \beta}} \mS^{ij,b}_\beta (\kappa_{\gamma, \beta} \mS^{ii,a}_\gamma - \kappa_{\beta, \gamma} \mS^{jj,a}_\gamma ) T_{\gamma + \beta} \varphi_\gamma (z-z_a, \omega_\gamma) \varphi_\beta (z-z_b, \omega_\beta + \frac{q_{ij} + \eta}{N}) + }
     \\ \ \\
    \displaystyle{+ \sum_{a,b=1}^n \sum_{\gamma, \beta} \mS^{ij,a}_\gamma (\kappa_{\gamma, \beta} \mS^{jj,b}_\beta - \kappa_{\beta, \gamma} \mS^{ii,b}_\beta ) T_{\gamma + \beta} \varphi_\gamma (z-z_a, \omega_\gamma + \frac{q_{ij}}{N}) \varphi_\beta (z-z_b, \omega_\beta + \frac{\eta}{N}) + }
    \end{array} \end{equation}
    $$
    \begin{array}{c}
    \displaystyle{+ \sum_{k: k\neq i,j}^M \sum_{a,b=1}^n \sum_{\gamma, \beta} T_\gamma T_\beta \mathcal{S}_\gamma^{ik, a} \mathcal{S}_\beta^{kj, b} \Big( \varphi_\gamma (z-z_a, \omega_\gamma + \frac{q_{ik}}{N}) \varphi_\beta (z-z_b, \omega_\beta + \frac{q_{kj} + \eta}{N})-}
     \\ \ \\
    \displaystyle{- \varphi_\gamma (z-z_a, \omega_\gamma + \frac{q_{ik} + \eta}{N}) \varphi_\beta (z-z_b, \omega_\beta + \frac{q_{kj}}{N}) \Big) +}
     \\ \ \\
    \displaystyle{+ \sum_{a=1}^n \sum_\gamma \mathcal S^{ij, a}_\gamma (\mu_i - \mu_j) T_\gamma f_\gamma (z-z_a, \omega_\gamma + \frac{q_{ij} + \eta}{N}).}
    \end{array}
    $$
It includes $\left[ \mathcal{L}, \mathcal{M} \right]^{ij}$ and the additional term (in the last line).
Again, the part of the terms with $a=b$ in all sums was derived in \cite{Reltops}. This part provides the upper line of \eqref{eqij}.
The term with the summation over $k$ is transformed through (\ref{Fay}). Then the expression under this sum takes the form:
\begin{equation} \begin{array}{c}
    \displaystyle{
    \mS^{ik, a}_\gamma \mS^{kj, b}_\beta T_{\gamma} T_{\beta} \varphi_{\gamma + \beta} (z-z_a, \omega_{\gamma +\beta} + \frac{q_{ij} + \eta}{N}) \Big( \varphi_\beta (z_{ab}, \omega_\beta + \frac{q_{kj} + \eta}{N}) - \varphi_\beta (z_{ab}, \omega_\beta + \frac{q_{kj}}{N})\Big) +
    }
     \\ \ \\
    \displaystyle{
     + \mS^{ik, b}_\beta \mS^{kj, a}_\gamma T_{\beta} T_{\gamma} \varphi_{\gamma + \beta} (z-z_a, \omega_{\gamma +\beta} + \frac{q_{ij} + \eta}{N}) \Big( \varphi_\beta (z_{ab}, \omega_\beta + \frac{q_{ik}}{N}) - \varphi_\beta (z_{ab}, \omega_\beta + \frac{q_{ik} + \eta}{N})\Big)\,.
    }
\end{array} \end{equation}
All computations are similar to the diagonal case, and in this way we get the last two terms in \eqref{eqij}. At the same time we also get additional terms, which are as follows:
\begin{equation} \begin{array}{c}
    \displaystyle{
    \sum_\gamma \sum_{a, b} \mS^{ij, a}_\gamma ( \mS^{ii,b}_{0,0} - \mS^{jj,b}_{0,0} ) T_\gamma \varphi_\gamma (z - z_a, \omega_\gamma + \frac{q_{ij} + \eta}{N}) (  E_1 (z_{ab}) - \phi(z_{ab}, \frac{\eta}{N})) +
    }
    \\ \ \\
    \displaystyle{
     + \sum_\gamma \sum_{a, b} \mS^{ij, a}_\gamma ( \mS^{ii,b}_{0,0} - \mS^{jj,b}_{0,0} ) T_\gamma f_\gamma (z - z_a, \omega_\gamma + \frac{q_{ij} + \eta}{N}).
     }
\end{array} \end{equation}
The upper expression provides the middle line in the equation (\ref{eqij}).
And the last one expression is cancelled together with the last terms in (\ref{w61}) and (\ref{w62}) by the definition \eqref{cond}.
\end{proof}

Finally, we mention that the Lax equation holds true on the constraints
  \beq\label{w64}
  \begin{array}{c}
    \displaystyle{
 \mu_i=0\,,\quad i=1,...,M\,,
 }
 \end{array}
 \eq
 where $\mu_i$ are given by (\ref{cond}). In this case the additional term in (\ref{nonLaxrel}) vanishes.
 By differentiating (\ref{w64}) with respect to time variable we find equations of motion for
 the positions of particles:
  \beq\label{w641}
  \begin{array}{c}
    \displaystyle{
 {\ddot q}_i=N\sum\limits_{a=1}^n \dot\mS^{ii,a}_{0,0}=\sum\limits_{a=1}^n \tr(\dot\mS^{ii,a})\,,\quad i=1,...,M\,,
 }
 \end{array}
 \eq
 Summing up (in $a$) equations (\ref{eqii}) and taking trace of both sides we find:
\begin{equation}\begin{array}{c}
\label{w642}
    \displaystyle{
    {\ddot q}_i =
     \sum\limits_{a=1}^n\sum_{k: k\neq i}^M \tr\Big( \mS^{ik,a} J^{\eta, q_{ki}} (\mS^{ki,a}) - J^{\eta, q_{ik}} (\mS^{ik,a}) \mS^{ki,a} \Big) +
     }
     \\ \ \\
     \displaystyle{
      +\sum_{k: k\neq i}^M \sum_{a,b: b\neq a}^n \tr\Big( \mS^{ik,a} \widetilde{J}^{\eta, q_{ki}}_a (\mS^{ki,b}) - \widetilde{J}^{\eta, q_{ik}}_a (\mS^{ik,b}) \mS^{ki,a}\Big)\,.
      }
\end{array}\end{equation}

 The constraints (\ref{cond}) should be also supplied with $M$ gauge fixation conditions thus performing (the Hamiltonian or the Poisson) reduction to the phase space of integrable model. The reduction is not only restriction of equations (\ref{eqij})-(\ref{eqii}) to the level of constraints (\ref{w64}) but provides some additional terms in the equations through the Dirac brackets formula. The same phenomenon takes place in the non-relativistic models. For example,
 the spin Calogero-Moser model with the spin variables from the minimal coadjoint orbit is reduced in this way to the spinless system (the model 9 on the Scheme 1).

\subsection{Particular cases}

\paragraph{Classical spin chain.} First of all we mention that in the $M=1$ case the Lax matrix (\ref{w56}), (\ref{relLaxl})
is simplified to the following one
  \beq\label{w644}
  \begin{array}{c}
    \displaystyle{
\mL(z) =
     \sum\limits_{a=1}^n\sum_{\gamma\in\mZ_N^{\times 2}}   T_\gamma
     \mS^{a}_\gamma \vf_\gamma(z-z_a, \om_\gamma + \frac{\eta}{N})\in{\rm Mat}(N,\mC)\,,
 }
 \end{array}
 \eq
 which is the monodromy matrix (\ref{w41}) of the classical chain in the additive form (with the redefinition $\eta\rightarrow \eta/n$).

\paragraph{${\rm GL}_{NM}$ model.}
When $n=1$ we have a single pole, which can be fixed as $z_1=0$. Then the Lax matrix (\ref{w56}), (\ref{relLaxl})
turns into
  \beq\label{w65}
  \begin{array}{c}
    \displaystyle{
\mL(z) =
     \sum\limits_{i,j=1}^M\sum_{\gamma\in\mZ_N^{\times 2}}  E_{ij}\otimes T_\gamma
     \mS^{ij}_\gamma \vf_\gamma(z, \om_\gamma + \frac{q_{ij} + \eta}{N})\in{\rm Mat}(NM,\mC)\,.
 }
 \end{array}
 \eq
Detailed description of this model (it is the model 4 on the Scheme 2) can be found in \cite{Reltops}.

Let us also mention several important particular cases of the ${\rm GL}_{NM}$ model itself. The first one is the model of $M$ interacting relativistic ${\rm GL}_N$ tops, which appears in the case ${\rm rank}(\mS)=1$. Recently a quantum version of this model was proposed in \cite{MZ}, and
related q-deformed long-rage spin chains were described. The second particular case is the relativistic top (\ref{w26}), which comes from (\ref{w65}) in the $M=1$ case. Finally, the third case is the spin Ruijsenaars-Schneider model \cite{KrichZ}, which corresponds to $N=1$. In our notation it is also briefly reviewed in the beginning of
\cite{Reltops}. If the matrix of spin variables has rank one then the reduction discussed in the end of the previous subsection kills all spin degrees of freedom and the spinless Ruijsenaars-Schneider model arises (see details in \cite{KrichZ}). All these
relations are shown on the Scheme 2.

\paragraph{Multispin Ruijsenaars model.} In the $N=1$ case the Lax matrix (\ref{w56}), (\ref{relLaxl}) becomes
the one for the multispin ${\rm GL}_M$ Ruijsenaars-Schneider model:
  \beq\label{w66}
  \begin{array}{c}
    \displaystyle{
\mL(z) =
     \sum\limits_{a=1}^n\sum\limits_{i,j=1}^M  E_{ij}
     \mS^{ij,a} \phi(z-z_a, q_{ij} + \eta)\in{\rm Mat}(M,\mC)\,.
 }
 \end{array}
 \eq
Similarly to transition between (\ref{w41}) and (\ref{w44}) the Lax matrix (\ref{w66}) can be transformed
to the form, which has no explicit dependence on the variable $\eta$\footnote{In order to get (\ref{w67}) from (\ref{w66}) one should divide $\mL(z)$ by function $\phi(z-x,\eta)$ for some $x$ and then represent the answer as in (\ref{w67}).} (see also the footnote to (\ref{w207})):
  \beq\label{w67}
  \begin{array}{c}
    \displaystyle{
\bar\mL_{ij}(z) =\delta_{ij}\Big(\bar\mS^{ii}+\sum\limits_{a=1}^{n'}\bar\mS^{ii,a}E_1(z-z_a)\Big)+
     (1-\delta_{ij})\sum\limits_{a=1}^{n'}
     \bar\mS^{ij,a} \phi(z-z_a, q_{ij})\,,\quad \sum\limits_{a=1}^{n'}\bar\mS^{ii,a}=0\,.
 }
 \end{array}
 \eq
 %
 This form is known for non-relativistic
multispin Calogero-Moser model introduced in \cite{N}. Also, this form was used in \cite{MMZ} for the multispin Ruijsenaars model. Although the form is non-relativistic, the Poisson structure is quadratic and complicated.

In fact, the Hamiltonian description is unknown even for $n=1$ case (the elliptic spin Ruijsenaars-Schneider model), but it is known for
$n=1$ and ${\rm rank}(\mS)=1$ since it is the spinless Ruijsenaars-Schneider model due to the additional reduction.
In the next Section we describe explicit parametrization in canonical variables of the model (\ref{w66}) with spin variables satisfying the property ${\rm rank}(\mS^k)=1$ for all $k=1,...,n$.


\section{Inhomogeneous Ruijsenaars chain}
\label{sec4}
\setcounter{equation}{0}

The Ruijsenaars chain on $n$ sites is the model introduced recently in \cite{ZZ}. Similarly to (\ref{w37}) it is described by $N\times N$
monodromy matrix
  \beq\label{e01}
  \begin{array}{c}
    \displaystyle{
 {\ti T}(z)={\ti L}^1(z)\ldots {\ti L}^n(z)\in\Mat\,,
 }
 \end{array}
 \eq
 where the Lax matrices ${\ti L}^1(z)$ are of the form
 \beq\label{e02}
 \begin{array}{c}
  \displaystyle{
 {\ti L}^k_{ij}(z)=
 \phi(z,{\bar q}^{k-1}_i-{\bar q}^{k}_j+\eta)
 \frac{\prod\limits_{l=1}^N\vth({\bar q}^k_j-{\bar q}^{k-1}_l-\eta) }
 {\vth(-\eta)\prod\limits_{l: l\neq j}^N\vth({\bar q}^{k}_j-{\bar q}^{k}_l) }\,e^{p^k_j/c}\,,
 }
 \end{array}
 \eq
 where $k=1,...,n$ and $i,j=1,...,N$. We deal here with $nN$ pairs of canonical variables:
 \beq\label{e03}
 \begin{array}{c}
  \displaystyle{
\{p_i^k,q_j^l\}=\delta^{kl}\delta_{ij}\,,\qquad
\{p_i^k,p_j^l\}=\{q_i^k,q_j^l\}=0\,,
 }
 \\ \ \\
  \displaystyle{
 i,j=1,...,N\,,\quad k,l=1,...,n\,.
 }
 \end{array}
 \eq
 Also, the cyclic identification $q^0_i=q^n_i$ is assumed. By construction, the transfer matrix ${\ti t}(z)=\tr {\ti T}(z)$ satisfies the involution property $\{{\ti t}(z),{\ti t}(w)\}=0$, thus providing Poisson commuting Hamiltonians.
 Likewise it happens in the classical homogeneous spin chains \cite{FTbook}, there is a flow in this model describing interaction of neighbour sites only. In continuous limit one obtains the local integrable 1+1 field theory -- the field generalization of the Ruijsenaars-Schneider model. See \cite{ZZ} for details.

 The model (\ref{e01}) is homogeneous, i.e. we put all $z_k=0$ in (\ref{w37}).
 In this subsection we describe a natural generalization of the above model to the inhomogeneous case. Then
 we discuss its relation to multispin Ruijsenaars model.

\subsection{Inhomogeneous Ruijsenaars chain}
The derivation of (\ref{e01})-(\ref{e02}) was based on consideration of the classical ${\rm GL}_N$ spin chain in the special case when all matrices of spin variables $S^1,...,S^n$ are of rank one. In this case the Lax matrices at each site are represented in the factorized form (\ref{w52}). Let us follow the same strategy in the
inhomogeneous case.

\paragraph{Lax matrices.} Consider the monodromy matrix (\ref{w37}) (with $\eta\rightarrow N\eta$)
  \beq\label{e031}
  \begin{array}{l}
  \displaystyle{
 T(z)=L^{N\eta}(S^1,z-z_1)L^{N\eta}(S^2,z-z_2)\ldots L^{N\eta}(S^n,z-z_n)
 }
 \end{array}
 \eq
and write all the Lax matrices in the form:
  \beq\label{e04}
  \begin{array}{l}
  \displaystyle{
 L^{N\eta}(S^k,z-z_k)=\frac{\vth'(0)}{\vth(\eta)}\,
 g(z-z_k+N\eta,q^k)\,e^{P^k/c}g^{-1}(z-z_k,q^k)\,.
 }
 \end{array}
 \eq
As we know from (\ref{w53}) the residue
  \beq\label{e05}
  \begin{array}{l}
  \displaystyle{
 S^k=\res\limits_{z=z_k}L^{N\eta}(S^k,z-z_k)
 }
 \end{array}
 \eq
 is explicitly expressed in terms of canonical variables $q^k_1,...,q^k_N$, $p_1^k,...,p_N^k$ and satisfy the
 classical Sklyanin algebra relations (\ref{w29}). Due to (\ref{e04}) the matrices $S^k$ have rank one (see \cite{VZ,ZZ}). Plugging (\ref{e05}) into the monodromy matrix (\ref{w37}) one gets
  \beq\label{e06}
  \begin{array}{l}
  \displaystyle{
 T(z)=\Big(\frac{\vth'(0)}{\vth(\eta)}\Big)^n
 g(z-z_1+N\eta,q^1)\,e^{P^1/c}g^{-1}(z-z_1,q^1)\ldots g(z-z_n+N\eta,q^n)\,e^{P^n/c}g^{-1}(z-z_n,q^n)\,.
 }
 \end{array}
 \eq
 Next, consider the gauge transformed monodromy matrix
  \beq\label{e07}
  \begin{array}{c}
  \displaystyle{
 {\ti T}(z)=G^{-1} T(z) G=
  }
  \\ \ \\
  \displaystyle{
 =\Big(\frac{\vth'(0)}{\vth(\eta)}\Big)^n
 g^{-1}(z-z_1,q^1)g(z-z_2+N\eta,q^2)\,e^{P^2/c}\ldots g^{-1}(z-z_n,q^n)g(z-z_1+N\eta,q^1)\,e^{P^1/c}\,,
 }
 \end{array}
 \eq
 where
  \beq\label{e071}
  \begin{array}{l}
  \displaystyle{
 G=g(z-z_1+N\eta,q^1)\,e^{P^1/c}\,.
 }
 \end{array}
 \eq
In this way we come to
  \beq\label{e072}
  \begin{array}{l}
  \displaystyle{
 {\ti T}(z)={\mathbb L}^1(z-z_1){\mathbb L}^2(z-z_2)\ldots {\mathbb L}^n(z-z_n)
 }
 \end{array}
 \eq
 with
  \beq\label{e08}
  \begin{array}{l}
  \displaystyle{
 {\mathbb L}^k(z-z_k)=\frac{\vth'(0)}{\vth(\eta)}\,g^{-1}(z-z_{k},q^{k})g(z-z_{k+1}+N\eta,q^{k+1})\,e^{P^{k+1}/c}\,.
 }
 \end{array}
 \eq
 Now we are in position to calculate (\ref{e08}). For this purpose we use the following formula proved in \cite{Hasegawa}:
 \beq\label{e09}
 \begin{array}{c}
  \displaystyle{
 \Big( -\vth'(0)\, g^{-1}(z,{ q}^{m})g(z+N\eta,{ q}^k)\Big)_{ij}=
 \phi(z,{\bar q}^{m}_i-{\bar q}^{k}_j+\eta)\,
 \frac{\prod\limits_{l=1}^N\vth({\bar q}^k_j-{\bar q}^{m}_l-\eta) }
 {\prod\limits_{l: l\neq j}^N\vth({\bar q}^{k}_j-{\bar q}^{k}_l) }\,.
 }
 \end{array}
 \eq
 By writing $g(z-z_{k+1}+N\eta,q^{k+1})$ as $g(z-z_{k}+N(\eta-\frac{z_{k+1}-z_{k}}{N}),q^{k+1})$
 and using (\ref{e09}) we obtain
  \beq\label{e10}
  \begin{array}{l}
  \displaystyle{
 {\mathbb L}^k_{ij}(z-z_k)=
 \phi\Big(z-z_k,{\bar q}^{k}_i-{\bar q}^{k+1}_j+\eta-\frac{z_{k+1}-z_{k}}{N}\Big)
 \frac{\prod\limits_{l=1}^N\vth({\bar q}^{k+1}_j-{\bar q}^{k}_l+\frac{z_{k+1}-z_{k}}{N}-\eta) }
 {\vth(-\eta)\prod\limits_{l: l\neq j}^N\vth({\bar q}^{k+1}_j-{\bar q}^{k+1}_l) }\,e^{p^{k+1}_j/c}\,.
 }
 \end{array}
 \eq
 Notice that enumeration of Lax matrices shifted by 1 with respect to the one used in the homogeneous case (\ref{e02}), i.e. when $z_1=...=z_n=0$ ${\mathbb L}^k(z-z_k)$ turns into ${\ti L}^{k+1}(z)$.
 This is just in order to achieve matching with numeration of poles $z_1,...,z_n$. In fact, it is more properly	to enumerate ${\mathbb L}$ (and $\ti L$) by two neighbour indices since each Lax matrix depends on two sets of variables $q$ with neighbour values of upper indices.

It is possible to slightly simplify expression (\ref{e10}) by introducing the variables
 \beq\label{e11}
 \begin{array}{c}
  \displaystyle{
 {\breve q}^k_i= {\bar q}^k_i+\frac{z_k}{N}\,,\quad k=1\ldots n\,,\quad i=1\ldots N\,.
 }
 \end{array}
 \eq
 Then the Lax matrix (\ref{e10}) takes the form:
  %
 \beq\label{e12}
 \begin{array}{c}
  \displaystyle{
 {\mathbb L}^k_{ij}(z-z_k)=
 \phi(z-z_k,{\breve q}^{k}_i-{\breve q}^{k+1}_j+\eta)
 \frac{\prod\limits_{l=1}^N\vth({\breve q}^{k+1}_j-{\breve q}^{k}_l-\eta) }
 {\vth(-\eta)\prod\limits_{l: l\neq j}^N\vth({\breve q}^{k+1}_j-{\breve q}^{k+1}_l) }\,e^{p^{k+1}_j/c}\,.
 }
 \end{array}
 \eq
 There is the following relation between (\ref{e12}) and (\ref{e02}):
 \beq\label{e13}
 \begin{array}{c}
  \displaystyle{
 {\mathbb L}^k_{ij}(z-z_k)={\tilde L}_{ij}^{k+1}(z-z_k)\Big|_{{\bar q}\rightarrow {\breve q}}\,.
 }
 \end{array}
 \eq

\paragraph{Non-local Hamiltonians.} Recall that by construction, the monodromy matrix (\ref{e072})
is gauge equivalent to the one for XYZ spin chain (\ref{e031}) in the special case (\ref{e04}). Therefore,
 \beq\label{e14}
 \begin{array}{c}
  \displaystyle{
 \tr {\ti T}(z)=\tr T(z)\,,
 }
 \end{array}
 \eq
 and the non-local Hamiltonians for the inhomogeneous Ruijsenaars chain
 \beq\label{e15}
 \begin{array}{c}
  \displaystyle{
 {\ti H}_i=\res\limits_{z=z_i}\tr {\ti T}(z)
 }
 \end{array}
 \eq
 are precisely the same as in the XYZ spin chain, see (\ref{w421}). It is possible to write these Hamiltonians more explicitly using (\ref{e12}). Introduce notations for $N\times N$ matrix (of Cauchy type)
 \beq\label{e16}
 \begin{array}{c}
  \displaystyle{
 C_{ij}(z,x,y)=\phi(z,x_i-y_j+\eta)\,,\quad i,j=1,...,N
 }
 \end{array}
 \eq
 and a set of $n$ diagonal $N\times N$ matrices
 \beq\label{e17}
 \begin{array}{c}
  \displaystyle{
 B^{[k,k+1]}_{ij}=\delta_{ij}\frac{\prod\limits_{l=1}^N\vth({\breve q}^{k+1}_j-{\breve q}^{k}_l-\eta) }
 {\vth(-\eta)\prod\limits_{l: l\neq j}^N\vth({\breve q}^{k+1}_j-{\breve q}^{k+1}_l) }\,e^{p^{k+1}_j/c}\,,\quad k=1,...,n\,,\quad i,j=1,...,N\,.
 }
 \end{array}
 \eq
 Then from (\ref{e072}) and (\ref{e12}) we conclude that
 \beq\label{e18}
 \begin{array}{c}
  \displaystyle{
 {\ti T}(z)=
 }
 \\ \ \\
   \displaystyle{
 =C(z-z_1,{\breve q}^1,{\breve q}^2)B^{[1,2]}C(z-z_2,{\breve q}^2,{\breve q}^3)B^{[2,3]}\ldots
 C(z-z_{n-1},{\breve q}^{n-1},{\breve q}^n)B^{[n-1,n]} C(z-z_n,{\breve q}^n,{\breve q}^1)B^{[n,1]}\,.
 }
 \end{array}
 \eq
 Due to (\ref{philimits}) we have
 \beq\label{e19}
 \begin{array}{c}
  \displaystyle{
 \res\limits_{z=z_i}C(z-z_i,x,y)=\rho\otimes\rho^T\,,\quad \rho=(1,1,...,1)^T\,,
 }
 \end{array}
 \eq
 where $\rho$ is a column-vector of units.
 Thus,
 \beq\label{e20}
 \begin{array}{c}
  \displaystyle{
 \res\limits_{z=z_i}{\ti T}(z)= C(z_i-z_1,{\breve q}^1,{\breve q}^2)B^{[1,2]}\ldots  C(z_i-z_{i-1},{\breve q}^{i-1},{\breve q}^i)B^{[i-1,i]}\times
 }
 \\ \ \\
   \displaystyle{
 \times\rho\otimes\rho^T B^{[i,i+1]} C(z_i-z_{i+1},{\breve q}^{i+1},{\breve q}^{i+2}) B^{[i+1,i+2]}\ldots C(z_i-z_n,{\breve q}^n,{\breve q}^1)B^{[n,1]}\,.
 }
 \end{array}
 \eq
After taking the trace, we finally get
 \beq\label{e21}
 \begin{array}{c}
  \displaystyle{
 {\ti H}_i=\rho^T B^{[i,i+1]} C(z_i-z_{i+1},{\breve q}^{i+1},{\breve q}^{i+2}) B^{[i+1,i+2]}\ldots C(z_i-z_n,{\breve q}^n,{\breve q}^1)B^{[n,1]}\times
 }
 \\ \ \\
   \displaystyle{
 \times
 C(z_i-z_1,{\breve q}^1,{\breve q}^2)B^{[1,2]}\ldots  C(z_i-z_{i-1},{\breve q}^{i-1},{\breve q}^i)B^{[i-1,i]}
 \rho\,.
 }
 \end{array}
 \eq

\subsection{Relation to multispin Ruijsenaars model}

Let us represent the monodromy matrix ${\ti T}(z)$ (\ref{e072}), (\ref{e12}) in the additive form similarly to what we did for spin chains, see (\ref{w37}), (\ref{w41}).

Due to (\ref{percond}) the Lax matrices (\ref{e12}) have the following quasi-periodic properties on the lattice of elliptic curve:
 \beq\label{e22}
 \begin{array}{c}
  \displaystyle{
 {\mathbb L}^k(z+1)={\mathbb L}^k(z)\,,
 }
 \\ \ \\
   \displaystyle{
 {\mathbb L}^k(z+\tau)=e^{-2\pi\imath\eta} H_k^{-1}{\mathbb L}^k(z)H_{k+1}\,,
 }
 \end{array}
 \eq
where $H_k\in\Mat$, $k=1,...,n$ are diagonal matrices
 \beq\label{e23}
 \begin{array}{c}
  \displaystyle{
  H_k={\rm diag}(e^{2\pi\imath {\breve q}_1^k},...,e^{2\pi\imath {\breve q}_N^k})\,.
 }
 \end{array}
 \eq
Therefore, for the monodromy matrix (\ref{e072}) we have
 \beq\label{e24}
 \begin{array}{c}
  \displaystyle{
 {\ti T}(z+1)={\ti T}(z)\,,
 }
 \\ \ \\
   \displaystyle{
 {\ti T}(z+\tau)=e^{-2\pi\imath n\eta} H_1^{-1} {\ti T}(z) H_{1}\,.
 }
 \end{array}
 \eq
The matrix ${\ti T}(z)$ has simple poles at $z_i$, $i=1,...,n$ with the residues given by the r.h.s. of (\ref{e20}).
Let us denote these residues as
 \beq\label{e25}
 \begin{array}{c}
  \displaystyle{
  \mS^i=\res\limits_{z=z_i}{\ti T}(z)\,.
 }
 \end{array}
 \eq
 In this way ${\ti T}(z)$ is fixed, and we can write it explicitly:
 \beq\label{e26}
 \begin{array}{c}
  \displaystyle{
  {\ti T}_{ij}(z)=\sum\limits_{k=1}^n \mS^k_{ij}\phi(z-z_k,q^1_i-q^1_j+n\eta)\,.
 }
 \end{array}
 \eq
 It is a general form of the ${\rm GL}_N$ multispin Ruijsenaars-Schneider model (\ref{w66}) with the substitution $\eta\rightarrow n\eta$.  Here we deal with
 a special set of the spin variables $\mS^i$ (\ref{e25}), (\ref{e20}), which are rank one matrices due to the presence
 of $\rho\otimes \rho^T$ in the product (\ref{e20}). Positions of particles are
 $q^1_1,...,q^1_N$, and the spin variables are parameterized by the rest of variables (i.e. by $q^k_1,...,q^k_N$ for $k=2,...,n$ and $p^m_1,...,p^m_N$, $m=1,...,n$).

 In $n=1$ case the obtained result (i.e. (\ref{e25})-(\ref{e26}) with (\ref{e20})) reproduces relation between
 the spinless Ruijsenaars-Schneider model (\ref{w45}) and the  relativistic top (\ref{w521}). The latter means that
 we obtained the parametrization of the reduced model (here we mean the reduction discussed in the end of subsection \ref{sec31}).


\section{Generalized model through $R$-matrix formulation}
\label{sec5}
\setcounter{equation}{0}

In this Section we extend the formulation of the most general ${\rm GL}_{NM}$ model presented in Section \ref{sec3}.
For example, it was mentioned in Section \ref{sec2} that the Lax matrix (\ref{w26}) of the relativistic top
is equivalently written in the form (\ref{w261}) with the Baxter-Belavin $R$-matrix (\ref{BB}). Being written in the $R$-matrix form
the model (\ref{w261}) is an extension of (\ref{w26}) since not only the elliptic $R$-matrix (\ref{BB}) can be used, but also any trigonometric or rational degeneration satisfying the associative Yang-Baxter equation with certain additional properties, which are the classical limit (\ref{serRz}), the unitarity (\ref{unitarity}) and the skew-symmetry (\ref{w031}). As a result, one obtains a more universal formulation of the Lax pair (\ref{w261}), (\ref{w35}) and the equations of motion (\ref{w31}), (\ref{w36}). Then the models are enumerated by possible $R$-matrices from a special but a wide class, which was briefly reviewed in our previous paper \cite{TZ}, and we do not repeat it here. The calculations below are based on the above mentioned $R$-matrix properties and identities
from the Appendix B.

Introduce the following  Lax pair:
\begin{equation}\label{e501}
\begin{array}{c} \displaystyle{
    \mL(z) = \sum_{i,j=1}^M E_{ij} \otimes \mL^{ij}(z) \in \text{Mat}(NM, \mathbb{C}), \quad \mL^{ij}(z) \in \text{Mat}(N, \mC).} \\ \ \\
    \displaystyle{ \mL^{ij}(z) = \sum_{a=1}^n \text{tr}_2 ( R^{z-z_a}_{12} (q_{ij} + \eta) P_{12} \mS^{ij,a}_2 ),\quad q_{ij} = q_i - q_j}
\end{array}\end{equation}
and
\begin{equation}\label{e502}
\begin{array}{c}
  \displaystyle{
    \mM(z) = \sum_{i,j=1}^M E_{ij} \otimes \mM^{ij}(z) \in \text{Mat}(NM, \mathbb{C}), \quad \mM^{ij}(z) \in \text{Mat}(N, \mC)\,,}
    \\ \ \\
    \displaystyle{ \mM^{ij}(z) = - \delta_{ij} \sum_{a=1}^n \text{tr}_2 ( R^{z-z_a, (0)}_{12} P_{12} \mS^{ii,a}_2 ) - (1 - \delta_{ij}) \sum_{a=1}^n \text{tr}_2 ( R^{z-z_a}_{12} (q_{ij}) P_{12} \mS^{ij,a}_2 )\,. }
\end{array}\end{equation}
Define also the set of linear operators (the extensions of (\ref{inerten1})-(\ref{inerten4})):
\begin{equation}
\label{inertenR1}
    \widetilde{\mathcal{J}}^{\eta, q_{mn}}_a (\mS^{ij,b}) = \text{tr}_2 \Big( \Big( R^{z_{ab}}_{12} (q_{mn} + \eta) - R^{z_{ab}}_{12} (q_{mn}) \Big) P_{12} \mS^{ij,b}_2 \Big)\,,
\end{equation}
\begin{equation}
\label{inertenR2}
    \widetilde{\mathcal{J}}^{\eta}_a (\mS^{ij,b}) = \text{tr}_2 \Big( \Big( R^{z_{ab}}_{12} (\eta) - R^{z_{ab}, (0)}_{12} \Big) P_{12} \mS^{ij,b}_2 \Big)\,,
\end{equation}
\begin{equation}
\label{inertenR3}
    \mathcal{J}^{\eta, q_{mn}}(\mS^{ij,b}) = \text{tr}_2 \Big( \Big( R^{q_{mn} + \eta, (0)}_{12}  - R^{q_{mn}, (0)}_{12} \Big) \mS^{ij,b}_2 \Big)\,,
\end{equation}
\begin{equation}
\label{inertenR4}
    \mathcal{J}^\eta(\mS^{ij,b}) = \text{tr}_2 \Big( \Big( R^{\eta, (0)}_{12}  - r^{(0)}_{12} \Big) \mS^{ij,b}_2 \Big)\,.
\end{equation}
The goal of this Section is to derive equations of motion. For the non-diagonal blocks ($i\neq j$) of $\mS$ matrix we have:
\begin{equation}\begin{array}{c}
\label{eqmRij}
    \displaystyle{\Dot{ \mS }^{ij,a} = \mS^{ij,a} \mathcal{J}^\eta (\mS^{jj, a}) - \mathcal{J}^\eta ( \mS^{ii,a} ) \mS^{ij,a} + \sum_{k:k\neq j}^M \mS^{ik,a} \mathcal{J}^{\eta, q_{kj}} (\mS^{kj,a}) - \sum_{k:k\neq i}^M \mathcal{J}^{\eta, q_{ik}} (\mS^{ik,a}) \mS^{kj,a} + } \\ \ \\
    \displaystyle{ + \sum_{b: b\neq a}^n \Big(\mS^{ij,a} \widetilde{\mathcal{J}}^\eta_a (\mS^{jj,b}) - \widetilde{\mathcal{J} }^{\eta}_a ( \mS^{ii,b} ) \mS^{ij,a} \Big) + \sum_{b: b\neq a}^n \Big( \sum_{k:k\neq j}^M \mS^{ik,a} \widetilde{\mathcal{J}}^{\eta, q_{kj}}_a (\mS^{kj, b}) - \sum_{k:k\neq i}^M \widetilde{\mathcal{J}}^{\eta, q_{ik}}_a (\mS^{ik, b}) \mS^{kj,a} ) \Big)\,. }
\end{array}\end{equation}
For the diagonal blocks equations are as follows:
\begin{equation}\begin{array}{c}
\label{eqmRii}
    \displaystyle{\Dot{ \mS }^{ii,a} = [\mS^{ii,a}, \mathcal{J}^\eta (\mS^{ii, a})] + \sum_{k: k\neq i}^M \Big( \mS^{ik,a} \mathcal{J}^{\eta, q_{ki}} (\mS^{ki,a}) - \mathcal{J}^{\eta, q_{ik}} (\mS^{ik,a}) \mS^{ki,a} \Big) + } \\ \ \\
    \displaystyle{ + \sum_{b: b\neq a}^n [\mS^{ii,a}, \widetilde{\mathcal{J}}^\eta_a (\mS^{ii,b})] + \sum_{b: b\neq a}^n \sum_{k: k\neq i}^M \Big( \mS^{ik,a} \widetilde{\mathcal{J}}^{\eta, q_{ki}}_a (\mS^{ki, b}) - \widetilde{\mathcal{J}}^{\eta, q_{ik}}_a (\mS^{ik, b}) \mS^{ki,a} ) \Big). }
\end{array}\end{equation}

\begin{theorem}
The equations of motion \eqref{eqmRij} and \eqref{eqmRii} are equivalent to the Lax equation with additional term:
\begin{equation}
\label{nonLaxR}
    \frac{d}{dt}\, \mL(z) = [ \mL(z), \mM(z) ] +
     \sum_{i, j=1}^M \sum_{a=1}^n E_{ij}\otimes \text{tr}_2 \Big( (\mu_i - \mu_j) F_{12}^{z-z_a} ( q_{ij} + \eta ) P_{12} \: \mS^{ij,a}_2 \Big)
\end{equation}
for matrices \eqref{e501} and \eqref{e502}, where:
\begin{equation}
\label{LaxCond}
    \mu_i = \Dot{q}_i - N \sum_{a=1}^n \mS^{ii,a}_{0,0}=
    \Dot{q}_i - \sum_{a=1}^n \tr(\mS^{ii,a})\,, \quad i=1, \ldots, M\,.
\end{equation}
\end{theorem}
\begin{proof} The structure of the proof is similar to the proof of Theorem \ref{th1}. Consider the l.h.s.
of (\ref{nonLaxR}):
\begin{equation}
\label{leftLaxR}
    \frac{d \mathcal{L}}{dt} = \sum_{i, j} E_{ij} \otimes \sum_{a=1}^n \text{tr}_2 \Big( R^{z-z_a}_{12} (q_{ij} + \eta) P_{12} \Dot{\mathcal{S}}^{ij, a}_2 + F^{z-z_a}_{12} (q_{ij} + \eta) P_{12} \mathcal{S}^{ij, a}_2 \Dot{q}_{ij} \Big)\,.
\end{equation}
For the r.h.s. of (\ref{nonLaxR}) we consider the diagonal and non-diagonal  blocks separately.
The diagonal part is as follows:
\begin{equation}\label{w70} \begin{array}{c}
    \displaystyle{
    [ \mathcal{L}, \mathcal{M} ]^{ii} = \sum_{a, b=1}^n \text{tr}_{23} \Big( \Big(R^{z-z_a, (0)}_{12} R^{z-z_b}_{23} (\eta) - R^{z-z_a}_{12} (\eta) R^{z-z_b, (0)}_{23} \Big) P_{12} P_{13} \mathcal{S}^{ii, a}_2 \mathcal{S}^{ii, b}_3 \Big) +
    }
    \\ \ \\
    \displaystyle{
    + \sum_{a, b=1}^n \sum_{k: k\neq i}^M \text{tr}_{23} \Big( \Big(R^{z-z_a}_{12} (q_{ik}) R^{z-z_b}_{23} (q_{ki} + \eta) - R^{z-z_a}_{12} (q_{ik} + \eta) R^{z-z_b}_{23} (q_{ki}) \Big) P_{23} P_{12} \mathcal{S}^{ik, a}_2 \mathcal{S}^{ki, b}_3 \Big)\,.
    }
\end{array} \end{equation}
The calculations for the terms with $a=b$ are performed in the same way as was described in \cite{ReltopsR}. So that, here we consider the case $a \neq b$ only.
Using \eqref{fla1} and \eqref{fla2} we get the following expression for the upper line in the r.h.s. of (\ref{w70}):
\begin{equation} \begin{array}{c}
    \displaystyle{
    \sum_{a, b: a\neq b}^n \text{tr}_{23} \Big( R^{z-z_a}_{12} (\eta) P_{12} \Big( R^{z_{ab}, (0)}_{23} - R^{z_{ab}}_{23} (\eta) \Big) P_{23} \mathcal{S}^{ii, a}_2 \mathcal{S}^{ii, b}_3 \Big) +
     }
     \\ \ \\
     \displaystyle{ + \sum_{a, b: a\neq b}^n \text{tr}_{23} \Big( R^{z-z_a}_{12} (\eta) P_{12} \mathcal{S}^{ii, a}_2 \Big( R^{z_{ab}}_{23} (\eta) - R^{z_{ab}, (0)}_{23} \Big) P_{23} \mathcal{S}^{ii, b}_3 \Big)\,.
     }
\end{array} \end{equation}
Similarly, for the lower line in the r.h.s. of (\ref{w70}) one obtains
\begin{equation} \begin{array}{c}
    \displaystyle{ \sum_{a, b: a\neq b}^n \sum_{k:k\neq i}^M \text{tr}_{23} \Big( R^{z-z_a}_{12} (\eta) P_{12} \Big( R^{z_{ab}}_{23} (q_{ik}) - R^{z_{ab}}_{23} (q_{ik} + \eta) \Big) P_{23} \mathcal{S}^{ki, a}_2 \mathcal{S}^{ik, b}_3 \Big) + }
    \\ \ \\
    \displaystyle{ + \sum_{a, b: a\neq b}^n \sum_{k:k\neq i}^M \text{tr}_{23} \Big( R^{z-z_a}_{12} (\eta) P_{12} \mathcal{S}^{ik, a}_2 \Big( R^{z_{ab}}_{23} (q_{ki} + \eta) - R^{z_{ab}}_{23}(q_{ki}) \Big) P_{23} \mathcal{S}^{ki, b}_3 \Big)\,.
    }
\end{array} \end{equation}
Comparing the resulting expressions with the left side of the Lax equation \eqref{leftLaxR} and taking into account $a=b$ case, we get precisely the equation of motion for the diagonal blocks \eqref{eqmRii}.

Consider now the non-diagonal blocks in the r.h.s. of (\ref{nonLaxR}):
\begin{equation} \begin{array}{c} \label{nondiagLaxR}
    \displaystyle{[ \mathcal{L}, \mathcal{M} ]^{ij} = \sum_{a, b=1}^n \text{tr}_{23} \Big( \Big( R^{z-z_b, (0)}_{12} R^{z-z_a}_{23} (q_{ij} + \eta) - R^{z-z_b}_{12} (\eta) R^{z-z_a}_{23} (q_{ij}) \Big) P_{12} P_{13} \mathcal{S}^{ii, b}_2 \mathcal{S}^{ij, a}_3 \Big) + }\\ \ \\
    \displaystyle{ + \sum_{a, b=1}^n \text{tr}_{23} \Big( \Big( R^{z-z_a}_{12} (q_{ij}) R^{z-z_b}_{23} (\eta) - R^{z-z_a}_{12} (q_{ij} + \eta) R^{z-z_b, (0)}_{23} \Big) P_{12} P_{13} \mathcal{S}^{ij, a}_2 \mathcal{S}^{jj, b}_3 \Big) + } \\ \ \\
    \displaystyle{ + \sum_{a, b=1}^n \sum_{k:k\neq i,j}^M \text{tr}_{23} \Big( R^{z-z_a}_{12} (q_{ik}) R^{z-z_b}_{23} (q_{kj} + \eta) P_{12} P_{13} \mathcal{S}^{ik, a}_2 \mathcal{S}^{kj, b}_3 \Big)  - } \\ \ \\
    \displaystyle{ -  \sum_{a, b=1}^n \sum_{k:k\neq i,j}^M \text{tr}_{23} \Big( R^{z-z_a}_{12} (q_{ik} + \eta) R^{z-z_b}_{23} (q_{kj}) P_{12} P_{13} \mathcal{S}^{ik, a}_2 \mathcal{S}^{kj, b}_3 \Big).}
\end{array} \end{equation}
Again, we focus on the terms with $a\neq b$. The first two upper lines in (\ref{nondiagLaxR})
are transformed using \eqref{AYB} and \eqref{fla1}-\eqref{fla2}:
\begin{equation} \begin{array}{c}
    \displaystyle{ \sum_{a, b: a\neq b}^n \text{tr}_{23} \Big( R^{z-z_a}_{12} (q_{ij} + \eta) P_{12} \mathcal{S}^{ij, a}_2 \Big( R^{z_{ab}}_{23} (\eta) - R^{z_{ab},(0)}_{23} \Big) P_{23} \mathcal{S}^{jj, b}_3 \Big) - }
     \\ \ \\
    \displaystyle{ - \sum_{a, b: a\neq b}^n \text{tr}_{23} \Big( R^{z-z_a}_{12} (q_{ij} + \eta) P_{12} \Big( R^{z_{ab}}_{23} (\eta) - R^{z_{ab}, (0)}_{23} \Big) P_{23} \mathcal{S}^{ii, b}_3 \mathcal{S}^{ij, a}_2 \Big) + } \\ \ \\
    \displaystyle{ + \sum_{a, b: a\neq b}^n \text{tr}_{23} \Big( R^{z-z_a}_{12} (q_{ij} + \eta) P_{12} \mathcal{S}^{ii, a}_2 \Big( R^{z_{ab}}_{23} (q_{ij} + \eta) - R^{z_{ab}}_{23} (q_{ij}) \Big) P_{23} \mathcal{S}^{ij, b}_3 \Big) - }
     \\ \ \\
    \displaystyle{ - \sum_{a, b: a\neq b}^n \text{tr}_{23} \Big( R^{z-z_a}_{12} (q_{ij} + \eta) P_{12} \Big( R^{z_{ab}}_{23} (q_{ij} + \eta) - R^{z_{ab}}_{23} (q_{ij}) \Big) P_{23} \mathcal{S}^{ij, b}_3 \mathcal{S}^{jj, a}_2 \Big) + }
     \\ \ \\
    \displaystyle{ + \sum_{a, b: a\neq b}^n N \text{tr}_{2} \Big( F^{z-z_a}_{12} (q_{ij} + \eta) P_{12} (\mathcal{S}^{ii, b}_{0,0} - \mathcal{S}^{jj, b}_{0,0}) \mathcal{S}^{ij, a}_2 \Big)\,.}
\end{array} \end{equation}
Two lower lines in (\ref{nondiagLaxR})  are transformed through \eqref{AYB}. This yields:
\begin{equation} \begin{array}{c}
    \displaystyle{ \sum_{a, b: a\neq b}^n \sum_{k:k\neq i,j}^M \text{tr}_{23} \Big( R^{z-z_a}_{12} (q_{ij} + \eta) P_{12} \mathcal{S}^{ik, a}_2 \Big( R^{z_{ab}}_{23} (q_{kj} + \eta) - R^{z_{ab}}_{23} (q_{kj}) \Big) P_{23} \mathcal{S}^{kj, b}_3 \Big) - }
     \\ \ \\
    \displaystyle{ - \sum_{a, b: a\neq b}^n \sum_{k:k\neq i,j}^M \text{tr}_{23} \Big( R^{z-z_a}_{12} (q_{ij} + \eta) P_{12} \Big( R^{z_{ab}}_{23} (q_{ik} + \eta) - R^{z_{ab}}_{23} (q_{ik}) \Big) P_{23} \mathcal{S}^{ik, b}_3 \mathcal{S}^{kj, a}_2 \Big)\,.}
\end{array} \end{equation}
Plugging this into \eqref{nonLaxR} and taking into account \eqref{LaxCond} we get  \eqref{eqmRij}.
\end{proof}

On the constraints $\mu_i=0$ the additional term in the r.h.s. of (\ref{nonLaxR}) vanishes, and we get the Lax equations. At the same time on the constraints $\mu_i=0$ we may easily deduce equations of motion for positions of particles:
\beq\label{w71}
\begin{array}{c}
    \displaystyle{
    {\ddot q}_i=\sum\limits_{a=1}^n\tr\Big(\Dot{ \mS }^{ii,a}\Big) =  \sum\limits_{a=1}^n\sum_{k: k\neq i}^M \tr\Big( \mS^{ik,a} \mathcal{J}^{\eta, q_{ki}} (\mS^{ki,a}) - \mathcal{J}^{\eta, q_{ik}} (\mS^{ik,a}) \mS^{ki,a} \Big) +
    }
     \\ \ \\
    \displaystyle{
     + \sum_{a,b: b\neq a}^n \tr\sum_{k: k\neq i}^M \Big( \mS^{ik,a} \widetilde{\mathcal{J}}^{\eta, q_{ki}}_a (\mS^{ki, b}) - \widetilde{\mathcal{J}}^{\eta, q_{ik}}_a (\mS^{ik, b}) \mS^{ki,a} ) \Big)\,.
     }
\end{array}
\eq

In order to reproduce the elliptic case one should use the elliptic $R$-matrix (\ref{BB}). Plugging (\ref{BB}) into the expressions \eqref{inertenR1}-\eqref{inertenR4}, we get (here we use identities \eqref{sum1}-\eqref{sum4}):
\begin{equation} \begin{array}{c}
    \displaystyle{ \widetilde{\mathcal{J}}^{\eta, q_{mn}}_a (\mS^{ij,b}) \to \widetilde{J}^{\eta, q_{mn}}_a (\mS^{ij,b}),}
     \\ \ \\
    \displaystyle{ \widetilde{\mathcal{J}}^{\eta}_a (\mS^{ij,b}) \to \widetilde{J}^{\eta}_a (\mS^{ij,b}) + \mathcal{S}^{ij,b}_{0,0} \phi(z_{ab}, \frac{\eta}{N}) 1_N - \mathcal{S}^{ij,b}_{0,0} E_1 (z_{ab}) 1_N,}
    \\ \ \\
    \displaystyle{\mathcal{J}^{\eta, q_{mn}}(\mS^{ij,b}) \to J^{\eta, q_{mn}}(\mS^{ij,b}),}
     \\ \ \\
    \displaystyle{\mathcal{J}^\eta(\mS^{ij,b}) \to J^\eta(\mS^{ij,b}) + \mathcal{S}^{ij,b}_{0,0} E_1 (\frac{\eta}{N}) 1_N .}
\end{array}\end{equation}
In this way equations of motion in the $R$-matrix description \eqref{eqmRij}-\eqref{eqmRii} turn into equations \eqref{eqij}-\eqref{eqii}.

\section{Appendix A: elliptic functions}\label{sec:A}
\def\theequation{A.\arabic{equation}}
\setcounter{equation}{0}

Using the theta-function
\beq\begin{array}{c} \displaystyle{
     \vartheta (z)=\vartheta (z|\tau) = -\sum_{k\in \mathbb{Z}} \exp \left( \pi i \tau (k + \frac{1}{2})^2 + 2\pi i (z + \frac{1}{2}) (k + \frac{1}{2}) \right)\,,\quad {\rm}Im(\tau)>0
}\end{array}\eq
define
the  Kronecker elliptic function
\beq\displaystyle{
\label{philimits}
    \phi(z, u) =
            \frac{\vartheta'(0) \vartheta (z + u)}{\vartheta (z) \vartheta (u)}\,,
            \qquad \res\limits_{z=0}\phi(z,u)=1\,.
}\eq
It has the properties:
\beq\begin{array}{c} \displaystyle{
    \phi(z, u) = \phi(u, z), \quad \phi (-z, -u) = - \phi(z, u).
}\end{array}\eq
Its derivative $f(z,u) = \partial_u \vf(z,u)$ is given by
\beq\begin{array}{c} \displaystyle{
\label{derphi}
    f(z, u) = \phi(z, u)(E_1(z + u) - E_1(u)), \quad f(-z, -u) = f(z, u)\,,
}\end{array}\eq
where\footnote{Functions $E_1$ and $E_2$ are called the Eisenstein functions (the first and the second respectively).}
\beq\begin{array}{c} \displaystyle{
\label{Elimits}
    E_1(z) =
            \partial_z \ln \vartheta(z)\,,
         \quad E_2(z) = - \partial_z E_1(z) = \wp(z) - \frac{\vartheta'''(0) }{3\vartheta'(0)}\,,
}\end{array}\eq
\beq\begin{array}{c} \displaystyle{
    E_1(- z) = -E_1(z)\,, \quad E_2(-z) = E_2(z)\,, \quad f(0, u) = - E_2(u)\,.
}\end{array}\eq
 The local expansions near $z=0$ is as follows:
\beq\begin{array}{c} \displaystyle{
\label{serphi}
    \phi(z, u) = \frac{1}{z} + E_1 (u) + z\rho (u) + O(z^2)\,,\quad \rho (z) = \frac{E^2_1(z) - \wp(z)}{2}\,,
}\end{array}\eq
\beq\begin{array}{c} \displaystyle{
\label{serE}
    E_1(z) = \frac{1}{z} + \frac{z}{3} \frac{\vartheta'''(0) }{\vartheta'(0)} + O(z^3)\,.
}\end{array}\eq
The following quasi-periodic properties (on the lattice of periods $1$ and $\tau$) hold:
\beq\begin{array}{c}
\label{percond}
    \displaystyle{ E_1 (z + 1) = E_1 (z), \quad E_1(z + \tau) = E_1(z) - 2\pi i, } \\ \ \\
    \displaystyle{ E_2 (z + 1) = E_2 (z), \quad E_2(z + \tau) = E_2(z),} \\ \ \\
    \displaystyle{ \phi (z + 1, u) = \phi (z, u), \quad \phi(z + \tau, u) = e^{- 2\pi i u} \phi(z, u), } \\ \ \\
    \displaystyle{ f (z + 1, u) = f (z, u), \quad f (z + \tau, u) = e^{-2\pi i u}(f(z, u) - 2\pi i \phi (z, u)).}
\end{array}\eq
 In calculations we use the addition formula
%
\beq \begin{array}{c} \label{Fay} \displaystyle{
    \phi(z_1, u_1) \phi(z_2, u_2) = \phi(z_1, u_1 + u_2) \phi(z_2 - z_1, u_2) + \phi(z_2, u_1 + u_2) \phi(z_1 - z_2, u_1)
}\end{array}\eq
and its degenerations:
\beq\label{derdif}\begin{array}{c} \displaystyle{
    f(z_1, u_1) \phi(z_2, u_2) - \phi(z_1, u_1) f(z_2, u_2) = \phi(z_2, u_1 + u_2) f(z_{12}, u_1) - \phi(z_1, u_1 + u_2) f(z_{21}, u_2),
}\end{array}\eq
\beq\label{a974}\begin{array}{c} \displaystyle{
    f(z, u_1) \phi(z, u_2) - \phi(z, u_1) f(z, u_2) = \phi(z, u_1 + u_2) (E_2(u_2) - E_2 (u_1)),
}\end{array}\eq
\beq\begin{array}{c} \displaystyle{
\label{diffsign}
    \phi(z, u) \phi(z, -u) = E_2(z) - E_2(u) = \wp (z) - \wp (u),
}\end{array}\eq
\beq\label{a16}\begin{array}{c} \displaystyle{
    \phi(z, u_1) \phi(z, u_2) = \phi(z, u_1 + u_2) (E_1 (z) + E_1 (u_1) + E_1 (u_2) - E_1 (z + u_1 + u_2)),
}\end{array}\eq
\beq\label{a17}\begin{array}{c} \displaystyle{
    \phi(z_1, u) \phi(z_2, u) = \phi(z_1 + z_2, u) (E_1 (z_1) + E_1 (z_2)) - f(z_1 + z_2, u).
}\end{array}\eq
\beq\begin{array}{c} \displaystyle{
\label{sigma1}
    \phi (z_1, u) \rho(z_2) - E_1 (z_2) f(z_1, u) + \phi(z_2, u) f(z_{12}, u) - \phi(z_1, u) \rho(z_{21})
     = -\frac{1}{2}\, \partial_u f(z_1, u) ,
}\end{array}\eq
\beq\begin{array}{c} \displaystyle{
\label{Ep}
    (E_1(u+v) - E_1(u) - E_1(v))^2 = \wp(u+v) +  \wp(u) + \wp(v),
}\end{array}\eq
\beq\begin{array}{c} \displaystyle{
\label{sigma2}
    \phi (z, u) \rho(z) - E_1 (z) f(z, u) - \phi(z, u) \wp(u) = -\frac{1}{2}\, \partial_u f(z, u).
}\end{array}\eq
For $\al = (\al_1, \al_2)\in \mZ_N \times \mZ_N$ define the following set of functions\footnote{The functions (\ref{varphi}) are basis  elements in the space of sections of the ${\rm End}(V)$ for a holomorphic vector bundle $V$ over elliptic curve of degree 1.}:
\beq\begin{array}{c} \displaystyle{
\label{varphi}
 \displaystyle{
    \vf_\al (z, \om_\al + u) = \exp (2 \pi i \frac{\al_2}{N} z) \phi (z, \om_\al + u), \quad \om_\al = \frac{\al_1 + \al_2 \tau}{N},
    }
}\end{array}\eq
\beq\begin{array}{c} \displaystyle{
\label{varf}
 \displaystyle{
    f_\al (z, \om_\al + u) = \exp (2 \pi i \frac{\al_2}{N} z) f (z, \om_\al + u),
    }
}\end{array}\eq
\beq\begin{array}{c} \displaystyle{
    \label{formulaf}
     \displaystyle{
    f_\al (z, \om_\al + u) = \partial_u \vf_\al (z, \om_\al + u) = \vf_\al (z, \om_\al + u) (E_1(z + \om_\al + u) - E_1(\om_\al + u)).}
}\end{array}\eq
The addition formulae are of the form:
\beq\begin{array}{c}
\label{formulaPhi}
 \displaystyle{
        \vf_\al (z_1, \om_\al + u_1) \vf_\be (z_2, \om_\be + u_2) = \vf_\al (z_1 - z_2, \om_\al + u_1) \vf_{\al + \be} (z_2, \om_{\al+\be} + u_1 + u_2) +}
         \\ \ \\
          \displaystyle{
        + \vf_\be (z_2 - z_1, \om_\be + u_2) \vf_{\al + \be} (z_1, \om_{\al + \be} + u_1 + u_2)\,.
        }
\end{array}\eq
 In particular,
\beq\begin{array}{c}
\label{sinfei}
 \displaystyle{
    \vf_\al(z-z_a, \om_\al) \vf_\be(z-z_b, \om_\be) =
  }
  \\ \ \\
   \displaystyle{
    =\vf_\al(z_{ba}, \om_\al) \vf_{\al + \be}(z - z_b, \om_{\al + \be}) + \vf_\be(z_{ab}, \om_\be) \vf_{\al + \be}(z-z_a, \om_{\al + \be})
    },
\end{array}\eq
and
\beq\begin{array}{c}
\label{formulaE}
\displaystyle{
        \vf_\al (z, \om_\al + u_1) \vf_\be (z, \om_\be + u_2) = \vf_{\al + \be}(z, \om_{\al + \be} + u_1 + u_2)  \times
         }
         \\ \ \\
         \displaystyle{
         \times\Big( E_1(z)+ E_1(\om_\al + u_1) + E_1(\om_\be + u_2) - E_1(z + \om_{\al + \be} + u_1 + u_2)\Big)\,,
         }
\end{array}\eq
and
\beq\begin{array}{c}
\label{formulaEf}
\displaystyle{
        \vf_\al (z_1, \om_\al + u) \vf_\al (z_2, \om_\al + u) =
        }
        \\ \ \\
         \displaystyle{
         = \vf_{\al}(z_1 + z_2, \om_\al + u)  ( E_1(z_1) + E_1(z_2)) - f_\al (z_1 + z_2, \om_\al + u)\,.
         }
\end{array}\eq
The following identity (the finite Fourier transformation) underlies the Fourier symmetry
(\ref{w33}):
\begin{equation}\begin{array}{c} \label{sum1}
    \displaystyle{\frac{1}{N} \sum_\al \kappa^2_{\al, \be} \varphi_\al (Nz, \om_\al + \frac{x}{N}) = \varphi_\be (x, \om_\be + z), \quad \forall \be \in \mZ_N \times \mZ_N}.
\end{array}\end{equation}
Its degenerations yield the relations (see \cite{Fourier} and the Appendix from \cite{ZZ} for details):
\begin{equation}\begin{array}{c}
    \displaystyle{\frac{1}{N}\, E_1 (z) + \frac{1}{N} \sum_{\al \neq 0} \kappa^2_{\al, \be} \varphi_\al (z, \om_\al) = E_1 (\om_\be + \frac{z}{N}) + 2 \pi i \partial_\tau \om_\be }.
\end{array}\end{equation}
\begin{equation}\begin{array}{c}
    \displaystyle{\frac{1}{N} \sum_{\al \neq 0} \kappa^2_{\al, \be} \Big(E_1 (\om_\al) + 2 \pi i \partial_\tau \om_\al \Big) = E_1(\om_\be) + 2 \pi i \partial_\tau \om_\be, \quad \be \neq 0}.
\end{array}\end{equation}
\begin{equation}\begin{array}{c} \label{sum4}
    \displaystyle{\frac{1}{N} \sum_{\al \neq 0} \Big(E_1 (\om_\al) + 2 \pi i \partial_\tau \om_\al \Big) = 0}.
\end{array}\end{equation}

\section{Appendix B: $R$-matrix properties}\label{sec:B}
\def\theequation{B.\arabic{equation}}
\setcounter{equation}{0}

\paragraph{The Baxter-Belavin elliptic $R$-matrix} \cite{Baxter2}:
For the elliptic $R$-matrix the following {\em matrix basis} in $\Mat$ is useful:
\beq\label{a971}\begin{array}{c} \displaystyle{
        T_\al = \exp \left( \al_1 \al_2 \frac{\pi i}{N} \right) Q^{\al_1} \Lambda^{\al_2}, \quad \al = (\al_1, \al_2)\in \mZ_N \times \mZ_N, \quad T_0=T_{(0,0)} = 1_N,
}\end{array}\eq
where $Q\in{\rm Mat}_N$ and $\Lambda\in{\rm Mat}_N$ are matrix generators of non-commutative torus (the finite-dimensional representation of the Heisenberg group):
\beq\label{a972}\begin{array}{c} \displaystyle{
    Q_{jk} = \de_{jk} \exp \left( \frac{2\pi i}{N} k\right), \quad \Lambda_{jk} = \de_{j-k+1 = 0\ \hbox{mod} N}, \quad Q^N = \Lambda^N =1_N\,.
}\end{array}\eq
 The {\em commutation relations} take the form:
 \beq\label{Tcond0}
 \begin{array}{c}
  \displaystyle{
    \exp \left( \al_1 \al_2 \frac{2\pi i}{N} \right) Q^{\al_1} \Lambda^{\al_2}=\Lambda^{\al_2} Q^{\al_1}\,,
}
\end{array}\eq
\beq\label{Tcond}\begin{array}{c} \displaystyle{
    T_\al T_\be = \ka_{\al, \be} T_{\al + \be}, \quad \ka_{\al, \be} = \exp \left( \frac{\pi i}{N}(\al_2 \be_1 - \al_1 \be_2) \right), \quad \ka_{\al, \al + \be} = \ka_{\al, \be}, \quad \ka_{- \al, \be} = \ka_{\be, \al},
}\end{array}\eq
\beq\label{TrT}\begin{array}{c} \displaystyle{
    \hbox{tr} (T_\al T_\be) = N \de_{\al + \be}, \quad \de_\al = \de_{\al_1, 0} \de_{\al_2, 0},
}\end{array}\eq
\beq\label{braketsT}\begin{array}{c} \displaystyle{
    [T_\al, T_\be] = (\ka_{\al, \be} - \ka_{\be, \al}) T_{\al+ \be} =
     2i \sin \left( \frac{\pi}{N}(\al_2 \be_1 - \al_1 \be_2) \right) T_{\al + \be}.
}\end{array}\eq
The elliptic Baxter-Belavin $R$-matrix is of the form:
\begin{equation}\label{BB}
\begin{array}{c}
    \displaystyle{R^{z}_{12} (x) = \frac{1}{N} \sum_\al \varphi_\al (x, \frac{z}{N} + \om_\al) T_\al \otimes T_{-\al}}\,.
\end{array}\end{equation}
The $R$-matrix has the following local expansion near $z = 0$ (it is the classical limit since $z$ here plays the role of the Planck constant):
\begin{equation}\begin{array}{c}
\label{serRz}
\displaystyle{
    R^{z}_{12}(x) = \frac{1}{z}\, 1_N \otimes 1_N + r_{12}(x) + z \: m_{12} (x) + O(z^2),
    }
\end{array}\end{equation}
where $r_{12}$ is the classical $r$-matrix satisfying the {\em classical Yang-Baxter equation}:
\begin{equation}\begin{array}{c}
\label{CYB}
    \displaystyle{[r_{12}, r_{13}] + [r_{12}, r_{23}] + [r_{13}, r_{23}] = 0, \quad r_{ij} = r_{ij}(q_i-q_j)\,.
    }
\end{array}\end{equation}
Plugging \eqref{serphi} into \eqref{serRz} one gets explicit expressions for classical $r$-matrix and the next order term (the $m$-matrix):
\begin{equation}
\label{w340}
\begin{array}{c}
    \displaystyle{
    r_{12} (x) = \frac{1}{N}\, E_1(x)\, 1_N \otimes 1_N + \frac{1}{N} \sum_{\al \neq 0} \varphi_\al (x, \om_\al) T_\al \otimes T_{-\al}
    },
\end{array}\end{equation}
\beq\label{w350}\begin{array}{c}
    \displaystyle{ m_{12} (x) = \frac{1}{N^2}\, \rho(x)\, 1_N \otimes 1_N + \frac{1}{N^2} \sum_{\al \neq 0} f_\al (x, \om_\al) T_\al \otimes T_{-\al}}\,.
\end{array}\eq
The local expansion of the classical $r$-matrix has the form:
\begin{equation}\begin{array}{c}
\label{CYB2}
 \displaystyle{
    r_{12}(x) = \frac{1}{x}\, P_{12} + r^{(0)}_{12} + x r^{(1)}_{12} (x) + O(x^2),
    }
\end{array}\end{equation}
where in the elliptic case
\begin{equation}\begin{array}{c}
    \displaystyle{ r^{(0)}_{12} = \frac{1}{N} \sum_{\al \neq 0} (E_1 (\om_\al) + 2 \pi i \partial_\tau \om_\al ) T_\al \otimes T_{-\al}}.
\end{array}\end{equation}
In (\ref{CYB2}) the matrix permutation operator $P_{12}$ appears. It is as follows:
\begin{equation} \begin{array}{c}
\label{perm}
     \displaystyle{
    P_{12}=\sum\limits_{k,l=1}^N E_{kl}\otimes E_{lk}=\frac{1}{N}\sum\limits_{\al\in\,\mZ_N\times\mZ_N}T_\al\otimes T_{-\al}\in\Mat^{\otimes 2}\,.
    }
\end{array} \end{equation}



\subsection*{R-matrix properties and identities}

Here we give some more properties for the elliptic $R$-matrix (\ref{BB}) and its degenerations.

Besides the classical limit (\ref{serRz}), which provides the local expansion of $R^{z}_{12}(x)$ near $z=0$, it is
also useful to consider its expansion near $x=0$:
\begin{equation}\begin{array}{c}
\label{serRx}
\displaystyle{
    R^{z}_{12}(x) = \frac{1}{x}\, P_{12} + R^{z, (0)}_{12} + x R^{z, (1)}_{12} + O(x^2)\,,
    \qquad \res\limits_{x=0}R^{z}_{12}(x) =P_{12}\,,
    }
\end{array}\end{equation}
where $P_{12}$ is the permutation operator (\ref{perm}).

The next, is the Fourier symmetry:
\begin{equation}\label{w33}\begin{array}{c}
    R^z_{12}(x) P_{12} = R^x_{12}(z)\,,
\end{array}\end{equation}
which can be viewed as matrix analogue of the trivial property $\phi(z,u)=\phi(u,z)$.

The skew-symmetry property (\ref{w031}) provides a set of relations for the coefficients of the expansions
(\ref{serRz}) and (\ref{CYB2}):
\begin{equation}\begin{array}{c}
    R^{z}_{12}(x) = - R^{-z}_{21}(-x), \quad r_{12} (z) = -r_{21} (-z), \quad r_{12}^{(0)} = -r_{21}^{(0)}, \quad m_{12}(z) = m_{21} (-z).
\end{array}\end{equation}
Using the symmetry (\ref{w33}) one can also derive a set of relations between the coefficients of
the expansions (\ref{serRz}) and (\ref{serRx}):
 \eqref{serRx}:
\begin{equation}\begin{array}{cc}
    R^{z, (0)}_{12} = r_{12}(z) P_{12}\,, & r_{12}^{(0)} = r_{12}^{(0)} P_{12}\,, \\ \ \\
    R^{z, (1)}_{12} = m_{12}(z) P_{12}\,, & r_{12}^{(1)} = m_{12}^{(0)} P_{12}\,.
\end{array}\end{equation}
The following notation is used for the $R$-matrix derivative with respect to spectral parameter:
\begin{equation}\begin{array}{c}
    F^{z}_{12} (q) = \partial_q R^{z}_{12} (q).
\end{array}\end{equation}
Degenerations of the associative Yang-Baxter equation \eqref{AYB} provide the following set of identities:
\begin{equation}\begin{array}{c} \label{fla1}
    R^{z-z_a}_{12}(x) R^{z-z_b, (0)}_{23} = R^{z-z_b}_{13}(x) R^{z_{ba}}_{12}(x) + R^{z_{ab}, (0)}_{23} R^{z-z_a}_{13}(x) + P_{23} F^{z-z_a}_{13}(x),
\end{array}\end{equation}
\begin{equation}\begin{array}{c} \label{fla2}
    R^{z-z_a,(0)}_{12} R^{z-z_b}_{23} (x) =  R^{z_{ab}}_{23} (x) R^{z-z_a}_{13} (x) + R^{z-z_b}_{13}(x) R^{z_{ba},(0)}_{12} + F^{z-z_b}_{13}(x) P_{12}.
\end{array}\end{equation}
Finally, we assume the following $R$-matrix traces:
\begin{equation}\begin{array}{c} \label{traceR}
    \tr_1 R^z_{12} (x) = \tr_2 R^z_{12} (x) = \phi(z/N, x) 1_N,
    \quad \tr_1\; r_{12} (x) = E_1 (x) 1_N, \quad \tr_1 \; m_{12} (x) = \frac{1}{N}\rho(x) 1_N\,.
\end{array}\end{equation}
In the elliptic case the latter simply follows from the definitions (\ref{BB}), (\ref{w340}) and (\ref{w350}).


\subsection*{Acknowledgments}
We are grateful to A. Zabrodin for useful discussions.

This work was supported by the Russian Science Foundation under grant no. 19-11-00062,\\ https://rscf.ru/en/project/19-11-00062/ .


\begin{small}

\end{small}

\end{document}